\documentclass[11pt]{article}
\usepackage[T1]{fontenc}

\usepackage{tgheros}
\usepackage{amsmath,amsfonts,amsthm,mathtools,color}
\usepackage{mathtools,color}
\usepackage{fullpage}
\usepackage{thmtools,thm-restate}
\usepackage[linesnumbered,noend,ruled,noline]{algorithm2e}
\usepackage{setspace}
\usepackage{tikz}
\renewcommand{\algorithmcfname}{MECHANISM}
\SetAlFnt{\small}
\SetAlCapFnt{\small}
\SetAlCapNameFnt{\small}
\SetAlCapHSkip{0pt}
\IncMargin{-\parindent}
\DeclarePairedDelimiter{\floor}{\lfloor}{\rfloor}
\usepackage{graphicx}
\usepackage{amsthm}
\usepackage{hyperref}
\usepackage{cleveref}
\usepackage[numbers]{natbib}
\setcitestyle{acmnumeric}
\DontPrintSemicolon

\SetCommentSty{mycommfont}
\SetArgSty{textnormal}
\usepackage{soul}
\usepackage[title]{appendix}
\usepackage{apxproof}
\usepackage{enumitem}

\usepackage{authblk}

\newcommand{\costi}[1][i]{c_{#1}}
\newcommand{\pricei}[1][i]{p_{#1}}

\newcommand{\upperB}[1]{\mathcal{U}}
\newcommand{\lowerB}[1]{\mathcal{L}}

\newtheorem{theorem}{Theorem}
\newtheorem{lemma}[theorem]{Lemma}
\newtheorem*{lemma*}{Lemma}
\newtheorem*{theorem*}{Theorem}
\newtheorem{corollary}{Corollary}

\title{Deterministic Budget-Feasible Clock Auctions\thanks{This project was supported in part by  by NSF grants CCF-2008280 and CCF-1755955. The authors would like to thank Eva Tardos and Georgios Amanatidis for helpful discussions during the early stages of this project.}}
\date{}

\begin{document}
\author[a]{Eric Balkanski\thanks{\texttt{eb3224@columbia.edu}, Industrial Engineering and Operations Research Department}}
\author[a]{Pranav Garimidi\thanks{\texttt{pg2682@columbia.edu}, Computer Science Department}}
\author[b]{Vasilis Gkatzelis\thanks{\texttt{gkatz@drexel.edu}, Computer Science Department}}
\author[b]{Daniel Schoepflin\thanks{\texttt{schoep@drexel.edu}, Computer Science Department}}
\author[b]{Xizhi Tan\thanks{\texttt{xizhi@drexel.edu}, Computer Science Department}}
\affil[a]{Columbia University}
\affil[b]{Drexel University}
\date{} 
\renewcommand\Affilfont{\itshape\small}
\renewcommand\Authands{, }

\newcommand{\bidders}{\mathcal{N}}

 \newcommand{\setn}{{[n]}}
 \newcommand{\B}{{B}}
\newcommand{\V}{{v}}
\newcommand{\M}{{\mathcal{M}}}
\newcommand{\bids}{{\mathbf b}}
\newcommand{\pay}{{\mathbf p}}
\newcommand{\costs}{{\mathbf c}}

\newcommand{\marginal}[2]{v(#1 ~|~ #2)}

\newcommand{\vali}[1][i]{v(\{#1\})}

\newcommand{\Scurr}{S_{t}}

\newcommand{\Sprev}{S_{t-1}}

\newcommand{\Sfinal}{S}

\newcommand{\OPT}{\texttt{OPT}} 

\newcommand{\gOPT}{\tilde{\texttt{OPT}}}

\newcommand{\finalround}{\hat{t}}
\newcommand{\thresholdround}{t^*}

\newcommand{\ALG}{\texttt{ALG}}

\newcommand{\argmax}{\text{argmax}}

\newcommand{\mechone}{\textsc{Iterative-Pruning}}

\newcommand{\mechtwo}{\textsc{Simultaneous-Iterative-Pruning}}

\newcommand{\unmax}{\textsc{UnconstrainedSubMax}}

\newcommand{\R}{R}

\newcommand{\firstset}{W_1}

\newcommand{\secondset}{W_2}

\newcommand{\thirdset}{W_3}

\newcommand{\lastbidder}{j^*}

\newcommand{\combinedset}{\firstset\cup \secondset}

\newcommand{\bfirstset}{W_1}
\newcommand{\bsecondset}{\overline{W}_2}

\newcommand{\rconstant}{\beta}

\newcommand{\firstconstant}{\alpha}

\maketitle
\begin{abstract}
We revisit the well-studied problem of budget-feasible procurement, where a buyer with a strict budget constraint seeks to acquire services from a group of strategic providers (the sellers). During the last decade, several strategyproof budget-feasible procurement auctions have been proposed, aiming to maximize the value of the buyer, while eliciting each seller's true cost for providing their service. These solutions predominantly take the form of randomized sealed-bid auctions: they ask the sellers to report their private costs and then use randomization to determine which subset of services will be procured and how much each of the chosen providers will be paid, ensuring that the total payment does not exceed the buyer's budget. Our main result in this paper is a novel method for designing budget-feasible auctions, leading to solutions that outperform the previously proposed auctions in multiple ways.

First, our solutions take the form of descending clock auctions, and thus satisfy a list of very appealing properties, such as obvious strategyproofness, group strategyproofness, transparency, and unconditional winner privacy; this makes these auctions much more likely to be used in practice. Second, in contrast to previous results that heavily depend on randomization, our auctions are deterministic. As a result, we provide an affirmative answer to one of the main open questions in this literature, asking whether a deterministic strategyproof auction can achieve a constant approximation when the buyer's valuation function is submodular over the set of services. In addition to this, we also provide the first deterministic budget-feasible auction that matches the approximation bound of the best-known randomized auction for the class of subadditive valuations. Finally, using our method, we improve the best-known approximation factor for monotone submodular valuations, which has been the focus of most of the prior work.

\end{abstract}

\thispagestyle{empty} 
\clearpage
\pagenumbering{arabic}

\section{Introduction}
 A decade ago, the seminal paper of \citet{singer2010budget} was the first to analyze the following important mechanism design problem: a buyer with a hard budget constraint, $B$, is looking to acquire some services (or goods) from a group of sellers, $\bidders$. The buyer has a value function $v(S)$ for receiving the services of each subset of sellers $S\subseteq \bidders$, and her goal is to maximize this value, but each seller $i\in \bidders$ has a \emph{private} cost $\costi$ for providing the service and would need to be compensated accordingly. The objective in this problem is to design a polynomial-time auction that determines which subset of services, $S$, the buyer should acquire and what payment $\pricei$ each seller $i\in S$ should receive, while ensuring \emph{budget feasibility}, i.e., $\sum_{i\in S}\pricei \leq B$, and \emph{strategyproofness}, i.e., that reporting their true costs is the optimal strategy for all sellers. The main result of \citet{singer2010budget} was a prior-free auction (i.e., one that has no prior information regarding the sellers' costs) that achieves a constant approximation of the optimal value when the buyer's valuation function is monotone submodular.

Since then, this problem has received a lot of attention due to its distinctive combination of practical and theoretical appeal. From a practical standpoint, budget-feasible procurement captures a multitude of application domains, ranging from crowdsourcing markets \citep{SM13, anari2014mechanism}, to influence maximization \citep{singer2012win} and data acquisition \citep{roth2012conducting,horel2014budget}. From a theoretical standpoint, this problem stands out because, unlike most mechanism design problems, budget feasibility imposes a non-trivial constraint on the payments that the mechanism can use, which introduces new challenges. A list of impressive results managed to overcome many of these challenges, leading to several prior-free budget-feasible auctions for instances where the buyer's valuations are additive~\citep{chen2011approximability,anari2014mechanism,gravin2019optimal}, monotone submodular~\citep{chen2011approximability,anari2014mechanism,JT21}, non-monotone submodular~\citep{amanatidis2019budget,bei2017worst}, and even subadditive~\citep{dobzinski2011mechanisms,bei2017worst}.

All these results take the form of sealed-bid auctions: the sellers are asked to reveal their private costs to the auctioneer, who then uses this information to decide the outcome. Although sealed-bid auctions are ubiquitous in the mechanism design literature, they have significant shortcomings. For example, they lack transparency, so the bidders need to trust that the auctioneer will not mishandle their private information and will faithfully implement the auction protocol. Furthermore, even if a sealed-bid auction is provably strategyproof, in practice bidders often lie to such auctions (see, e.g., \cite{KHL87}), partly because their strategyproofness may be hard to verify. Motivated by this discrepancy, \citet{li2017obviously} recently introduced a more demanding notion, known as \emph{obvious strategyproofness}. In an obviously strategyproof auction, the bidders can trivially verify that they cannot benefit by manipulating the auction, and experimental evidence verifies that their behavior in practice conforms with the rules of these auctions.
This also implies other sought-after incentive properties, such as weak group strategyproofness (i.e., no coalition of bidders can misreport collectively and all benefit).

In search for a practical alternative to sealed-bid auctions, \citet{milgrom2020clock} recently identified a particularly noteworthy class of obviously strategyproof auctions, known as \emph{clock auctions}.  
In contrast to sealed-bid auctions, budget-feasible clock auctions take place over multiple rounds: in each round they offer a price to each bidder and the bidders have the opportunity to reject the price offered to them and permanently exit the auction. The price offered to each bidder weakly decreases over time, and the auction terminates when the prices offered to the bidders that remain active add up to no more than the budget, at which point the buyer acquires the services of the active bidders at the last price that they were offered. Apart from being obviously strategyproof, these auctions are highly transparent and do not require that the bidder trust the auctioneer. Motivated by these highly appealing properties, in this work we set out to design budget-feasible clock auctions.

Another important limitation of the previously proposed budget-feasible mechanisms is that the vast majority of them heavily rely on randomization, 
making it unlikely that they would be used in practice: on one hand, bidders can find the notion of randomization and its impact confusing and, on the other hand, it can be hard to verify that the resulting outcome is indeed the product of the prescribed randomization
\citep{AME2019uncertainty,JM17}. Furthermore, the performance bounds of randomized mechanisms are guaranteed only in expectation rather than ex-post. 
However, the design of deterministic budget-feasible auctions for submodular valuations has remained elusive, with  \citet{amanatidis2019budget} pointing to the problem of ``obtaining deterministic, budget-feasible, O(1)-approximation mechanisms—or showing that they do not exist'' as the most intriguing open problem in this line of work.

\subsection{Our Results}

In this paper we propose a new method for designing budget-feasible auctions that simultaneously addresses many of the shortcomings of previous mechanisms. First, our results take the form of \emph{clock auctions}, thus avoiding the shortcomings of sealed-bid auctions and leading to solutions that are much more likely to be used in practice. Second, they are \emph{deterministic}, making them even more practical and guaranteeing their performance ex-post. Finally, our auctions \emph{either beat or match the best known approximation bounds} by any other polynomial time strategyproof auction.
The approximation guarantees that we achieve through our deterministic clock auctions are:
\begin{itemize}[leftmargin=*]
    \item \textbf{monotone submodular valuations (Section~\ref{sec:monotone}):} our deterministic approximation of $4.75$, improves the best known randomized approximation of $5$ by~\citet{JT21}. 
    \item \textbf{non-monotone submodular valuations (Section~\ref{sec:nonmonotone}):} our deterministic approximation of $64$, improves the best known randomized approximation of $505$ by~\citet{amanatidis2019budget}.
    \item \textbf{subadditive valuations (Section~\ref{sec:subadditive}):} our deterministic approximation of $O(\log(n)/\log\log(n))$ matches the best known randomized approximation by~\citet{bei2017worst} and improves the best known deterministic approximation of $O(\log^3(n))$ by~\citet{dobzinski2011mechanisms}.
\end{itemize}

As a corollary, we resolve the open question posed by \citet{amanatidis2019budget} in the affirmative, by providing the first deterministic mechanism that combines strategyproofness with a constant factor approximation for instances where the buyer valuations are submodular. In fact, rather than using a sealed-bid auction, we achieve this result using the more restrictive class of clock auctions.

Our method for designing these budget-feasible clock auctions proceeds by initially making a pessimistic estimate regarding the quality of the optimal solution and determining the first set of prices to offer to the bidders, aiming to achieve that pessimistic estimate. If some budget-feasible group of bidders that accepted their prices is sufficient to satisfy the pessimistic estimate, this group is temporarily set aside. Then, the estimate is updated to be slightly more ambitious, and a new set of prices is offered to bidders, except the ones that were set aside. If this more ambitious estimate is achieved by a new budget-feasible group of bidders, then the set-aside bidders are replaced by the new group, and the process continues until the auction reaches an estimate that it is unable to achieve. By setting aside the group of bidders that achieved the latest estimate, these auctions secure that estimate before attempting to reach a more ambitious estimate. Meanwhile, these estimates are used as a guide for gradually more demanding pricing, so the process that leads to the discovery of the final prices has a primal-dual flavor. When designing our auctions the pace at which we raise the estimate balances a subtle tradeoff between increasing it slowly enough to avoid overshooting the target value we wish to reach, yet quickly enough to limit the  loss from optimal bidders who are eliminated in each phase due to overlapping value with non-optimal bidders.  

\paragraph{Clock auctions and algorithms} All clock auctions automatically satisfy multiple very appealing properties, like obvious-strategyproofness, weak group strategyproofness, transparency, simplicity, and unconditional winner privacy, which are, in general, not satisfied by sealed-bid auctions~\citep{milgrom2020clock}.
An exciting implication is that this reduces the problem of designing practical auctions to the \emph{purely algorithmic problem} of designing price increase trajectories without worrying at all about incentives. For instance, note that this paper requires no proofs regarding strategyproofness: we just design algorithms that follow the format of clock auctions and analyze their worst-case approximation guarantees. 
In fact, \citet{milgrom2020clock} proved that clock auctions correspond to a specific class of backward greedy algorithms. Specifically, every budget-feasible clock auction corresponds to a multi-round greedy algorithm that, in each round, summarizes the ``attractiveness'' of each bid using a score that depends only on its cost, and then myopically eliminates the active bidder with the lowest score. This process continues until the algorithm terminates and accepts all remaining bidders. Therefore, designing clock auctions is equivalent to designing backward greedy algorithms.  

This connection between clock auctions and backward greedy algorithms illuminates some non-trivial design challenges that we had to overcome. Although there are several classic results that use forward greedy algorithms for approximately maximizing a submodular function\footnote{In fact, many of the known strategyproof budget-feasible auctions for submodular valuations closely resemble these classic results, and are based on forward greedy algorithms.}, we are not aware of any prior work that achieves comparable guarantees using backward greedy algorithms. A crucial difference is that forward greedy algorithms proceed by iteratively \emph{adding} bidders to a set of accepted bidders, $A$, and they can myopically decide which bidder $i$ to add based on the marginal increase in value relative to the set of accepted bidders up to this point, i.e., $v(A\cup\{i\})-v(A)$. On the other hand, backward greedy algorithms need to myopically \emph{eliminate} the ``least appealing'' bidders. The main challenge is that it is hard for backward greedy algorithms to gauge the marginal contribution of each bidder relative to the accepted bidders since it cannot foresee who the accepted bidders are going to be. The backward greedy algorithm corresponding to our clock auction  functions as follows: in each round we set a gradually more demanding threshold regarding what would make each bidder ``acceptable'' and then reject the first bidder who does not pass that threshold.

\paragraph{Overcoming the dependence on randomization.}
Prior work leverages randomization in two fundamental ways. First, many of the mechanisms gradually construct two (overlapping) subsets of bidders and then choose which one of these two sets will be accepted, uniformly at random. This guarantees that the expected value of the outcome will be at least half of the maximum value among the two subsets, and it, crucially, maintains the monotonicity of the allocation rule, thus not compromising the incentives of the mechanism (see, e.g.,~\citep{chen2011approximability,JT21} for more details). Second, some recent mechanisms randomly sample some of the bidders, and then use the values of the sampled bidders to determine an estimate of the optimal value~\citep{bei2017worst,amanatidis2019budget}. Then, using this estimate as a benchmark, they approach the non-sampled bidders and offer them take-it-or-leave-it prices. 
Deterministic auctions cannot use random sampling to estimate the optimal value, but our auctions overcome this issue by gradually adjusting the estimate until they reach a reasonable approximation.
Our results are in contrast to previous work that has uncovered simple instances where the performance of deterministic clock auctions is asymptotically different than that of randomized ones~\cite{dutting2017performance}.

\subsection{Related Work}\label{sec:relatedwork}

\paragraph{Budget-feasible clock auctions}
Prior to this work, there were only a few examples of budget-feasible clock auctions in the literature. 
\citet{ensthaler2014dynamic} and \citet{JM17} focused on the very special case where $v(S)=|S|$.
\citet{BKS12} and \citet{BH16} designed budget-feasible posted-price mechanisms, which are a special type of clock auction, but the former only obtained a $O(\log n)$ approximation for monotone submodular valuations and the latter considered a Bayesian setting where the costs are drawn from a prior distribution known to the auctioneer.
All other known budget-feasible mechanisms take the form of sealed-bid auctions, and the vast majority of these auctions cannot be implemented as clock auctions, with just a few exceptions~\citep{bei2017worst,gravin2019optimal,amanatidis2019budget}. The mechanisms in~\cite{bei2017worst,amanatidis2019budget} rely on randomization by using randomized sampling to learn the costs of some subset of the bidders and then use this information to determine posted prices for the remaining bidders.
We have also verified that the two auctions in~\citep{gravin2019optimal} can be implemented as clock auctions, but their guarantees are restricted to the case of additive valuations.

\paragraph{Other budget-feasible auctions}
Starting from the results of \citet{singer2010budget} several budget-feasible auctions have been proposed. For the special case of additive valuations, \citet{chen2011approximability} improved the approximation by providing a $2+\sqrt{2}$-approximate deterministic mechanism and a $3$-approximate randomized mechanism. They also proved a lower bound of $2$ for any randomized mechanism and a lower bound of $1 + \sqrt{2}$ for any deterministic one.\footnote{These lower bounds are the best known even for the much more general class of subadditive valuation functions.} \citet{gravin2019optimal} proved a matching upper bound for randomized mechanisms and additive valuations and also provided a $3$-approximate deterministic mechanism. 
For monotone submodular  functions, \citet{chen2011approximability}
gave 
a $7.91$-approximate randomized 
mechanism and \citet{JT21} further improved this result, obtaining a randomized mechanism that achieves the best known approximation of 5. In this paper we achieve an improved approximation of $4.75$, while also providing the first polynomial time auction to achieve any constant approximation for monotone submodular valuations without using randomization.

Beyond monotone submodular valuations, finding mechanisms with small constant approximation factors has proven more elusive.  For non-monotone submodular valuations, \citet{amanatidis2019budget} gave a $505$-approximate randomized 
mechanism. 
Using similar techniques  
\citet{bei2017worst} had previously managed to design a randomized mechanism obtaining a $768$-approximation for XOS valuations. 
For the class of subadditive valuations, \citet{dobzinski2011mechanisms} gave a $O(\log^2{n})$-approximate randomized mechanism and a $O(\log^3(n))$-approximate deterministic one. \citet{bei2017worst} improved the former by giving a $O(\log{n}/\log{\log{n}})$-approximate randomized mechanism.\footnote{The results for XOS and subadditive valuations are obtained in the \emph{demand oracle model}, not in the value oracle model our mechanisms for submodular valuations use.  It takes an exponential number of value queries in expectation to obtain a (randomized) $n^{1-\epsilon}$-approximation for XOS function maximization for any fixed $\epsilon > 0$ \citep[Theorem 6.2]{amanatidis2019budget}.} Section~\ref{sec:nonmonotone} improves the best known approximation for non-monotone submodular from 505 to 64, while providing the first deterministic auction to achieve any constant approximation for this class. Section~\ref{sec:subadditive} improves the best known deterministic approximation for subadditive valuations from $O(\log^3(n))$ to $O(\log{n}/\log{\log{n}})$, matching the best known randomized approximation for this class.

Some prior work has also designed mechanisms for the more tractable \emph{large-market} model, which assumes that every bidder represents a vanishing portion of the optimal value. This assumption sidesteps some of the main obstacles that arise in budget-feasible mechanism design and enables better approximation for additive~\citep{anari2014mechanism} and monotone submodular valuations~\citep{JT21}.

\paragraph{``Simple'' mechanisms and other clock auctions}  Our work also adds to the developing literature on simplicity in mechanism design (e.g., \citep{babaioff2014simple,rubinstein2016computational,chen2018complexity}). Even if a mechanism is strategyproof, it may not be readily used in practice \citep{ausubel2006lovely}, e.g., because the participants may not understand or trust that the mechanism is strategyproof~\citep{KHL87, li2017obviously}.  Designing simple mechanisms often requires that the algorithmic processes are straightforward so that the participants can understand them.  Clock auctions, however, present an extremely simple, even obviously strategyproof, interface to the bidders regardless of how sophisticated the algorithmic techniques computing the clock prices are.  
Clock auctions then provide a striking balance of algorithmic richness with practical applicability.  

Motivated by the highly desirable characteristics of clock auctions, there is a growing literature examining their performance in a variety of settings, including procurement settings without budget constraints \citep{kim2015welfare}, forward auction settings where the bidders have private values for being served and there is a publicly known constraint system over the sets of feasible bidders \citep{dutting2017performance,gkatzelis2017deferred}, and double auction settings where the auctioneer is interacting with both buyers and sellers \citep{dutting2017modularity, loertscher2020asymptotically}.

\section{Preliminaries}
We consider a procurement setting with a set $\bidders$ of $n$ bidders each of which is capable of providing some service to the auctioneer. 
Each bidder $i \in \bidders$ has a private cost $\costi \geq 0$ which indicates the minimum payment that $i$ would require in order to provide her service.  
The auctioneer has a budget $\B$ that they can spend on services, and a non-negative valuation function $\V$: $2^{\bidders} \rightarrow \mathbb{R}^{\geq 0}$ that defines the value the auctioneer receives from acquiring the services of the bidders in each subset $S \subseteq \bidders$.

We say that the valuation function of the auctioneer $\V$ is \textit{monotone} if $\V(S) \leq \V(T)$ for any $S \subseteq T \subseteq \bidders$ and \textit{submodular} if  $\V(S)+\V(T)\geq \V(S\cap T) + \V(S \cup T)$  for all $S, T \subseteq \bidders$.  An equivalent definition of submodularity  is that a function $\V$ is submodular  if it satisfies the following diminishing returns property: $\V(X \cup \{i\}) - \V(X) \geq \V(Y \cup \{i\}) - \V(Y)$ for all sets $X \subseteq Y \subset \bidders$ and all $i \notin Y$.  We say a valuation function is \textit{subadditive} if $\V(S \cup T) \leq \V(S) + \V(T)$ for all $S, T \subseteq \bidders$. 
For any sets $S,T \subseteq \bidders$, we denote the \textit{marginal contribution} of $T$ when added to $S$ as $\marginal{T}{S} = v(S \cup T) - v(S)$.  

A \emph{(descending) clock auction} offers a sequence of non-increasing prices computed using only public information, one in each phase of the auction, to bidders. In other words, let $p_{i,t}$ denote the price the mechanism offers to bidder $i$ in phase $t$. Then, we have $p_{i,t} \leq p_{i,t-1}$ for all bidders $i$ and phases $t$. Upon receiving their offer, each bidder may choose to exit or continue the auction. Bidders who choose to continue are said to ``accept'' the lower price and are called \emph{active} bidders. We denote the set of active bidders at the end of phase $t$ as $A_t$ with $A_t \subseteq A_{t-1} \subseteq \dots \subseteq A_1 \subseteq \bidders$.  When the auction ends in phase $\hat{t}$, some subset of the active bidders is selected as the winning set $W$ and the service of each $i \in W$ is acquired at her most recently accepted price.
An auction that chooses a winning set $W$ and charges each $i\in W$ a price $p_{i,\hat{t}}$ is \emph{budget feasible} if $\sum_{i \in W} p_{i,\hat{t}} \leq \B$.  

We measure the performance of our mechanisms by comparing them against the optimal value achievable by a computationally unbounded auctioneer that also knows every bidder's true cost. If the auctioneer knew the private costs $\costs = (c_i)_{i \in \bidders}$ of the bidders, she would be able to select the subset $W$ of services with the maximum total value under the budget constraint, paying each bidder $i$ a price $p_i = c_i$. For some instance $I$, let $\mathcal{O}(I)$ denote the optimal set of bidders to be served in instance $I$ and $\OPT$ denote $v(\mathcal{O}(I))$.  Similarly, let $\M(I)$ denote the set of bidders served in instance $I$ by some mechanism $\mathcal{M}$.  We say that $\mathcal{M}$ achieves an approximation factor $\rho\geq 1$ for a class of instances $\mathcal{I}$ if it always extracts at least a $1/\rho$ fraction of the optimal value, i.e., $\rho \geq \sup_{I\in\mathcal{I}} \frac{\OPT}{v(\mathcal{M}(I))}.$

\section{Monotone Submodular Valuations}\label{sec:monotone}
In this section we develop a deterministic clock auction that achieves a $4.75$ approximation for any monotone submodular valuation function.  This is the first deterministic strategyproof budget-feasible mechanism that achieves a constant approximation for monotone submodular valuation functions  in polynomial time, and we achieve this with a clock auction. At the core of our clock auction is a novel backward greedy technique for maximizing submodular functions that iteratively eliminates bidders from consideration.

Our clock auction, called \mechone,  proceeds in phases that iteratively eliminate bidders. In each phase $t$ we aim to find a set of bidders $\Scurr$ with value at least  $\gOPT$, where $\gOPT$ is initially a low,  rough estimate of the optimal value $\OPT$ that is then gradually increased and  refined. At each phase, the mechanism iteratively considers the remaining bidder $i$ with maximum marginal contribution $\marginal{\{i\}}{\Scurr}$ to $\Scurr$. It then offers price $p_i = \min\{p_i,~ \marginal{\{i\}}{S_{t}} \cdot \frac{B}{\gOPT}\}$ to bidder $i$, which is the minimum of the last price offered to bidder $i$ and the marginal contribution of $i$ to the bidders $\Scurr$, scaled in order to reach the target value $\gOPT$ with budget $B$. If bidder $i$ accepts, the mechanism  adds $i$ to $\Scurr$, otherwise it eliminates $i$ from the set of active bidders $A$. Phase $t$ terminates either when $v(\Scurr)\geq \gOPT$, or when there are no more bidders to offer a price to, i.e., $A \setminus (S_{t-1} \cup S_{t})=\emptyset$.

At the beginning of a new phase $t>1$, the target $\gOPT_t$ is updated to be two times the previous target $\gOPT_{t-1}$. We set aside $\Sprev$, the bidders who accepted the price they were offered in the previous phase, $t-1$.
If  all the active bidders are either in $\Sprev$ or $\Scurr$ at the end of a phase $t$, then $t$ is the last phase of the mechanism. We implement sets in our mechanisms as lists that maintain the order in which the bidders were added to them.  We say that the \emph{prefix} of length $k$ of set $S$ is the subset comprising bidders from the first to the $k$-th index of the list representing $S$.
\vspace{7pt}

\begin{algorithm}[H]
\setstretch{1.1}
\SetKwInOut{Input}{Input}
\Input{Budget $B$,   valuation function $v: 2^N \rightarrow \mathbb{R}$}
 initialize  $A \leftarrow \bidders$, $S_{0} \leftarrow \emptyset, S_1 \leftarrow \left\{\text{argmax}_{i \in \bidders} \ \vali\right\}$,  $\gOPT \leftarrow v(S_1)$, $t \leftarrow 1$,  $p_i \leftarrow B$ for all $i \in \bidders$
 
\While{$A \setminus \left(S_{t-1} \cup S_{t}\right) \neq \emptyset$}{

update $t \leftarrow t + 1$, $\gOPT \leftarrow 2\gOPT$ and initialize $S_t \leftarrow \emptyset$ 
 \tcp*{start  a new phase} \label{algline:startouter}
 
 \While{$v(S_{t}) < \gOPT$ and $A \setminus \left(S_{t-1} \cup S_{t}\right) \neq \emptyset$}{
 
    let $i \leftarrow \text{argmax}_{i \in A \setminus \left(S_{t-1} \cup S_{t}\right)}\marginal{\{i\}}{ S_{t}}$ \;
  
  update $p_i \leftarrow \min\left\{p_i,~ \marginal{\{i\}}{S_{t}} \cdot \frac{B}{\gOPT}\right\}$
  
	\If{bidder $i$  accepts price $p_i$}{
	 update $S_{t} \leftarrow S_{t} \cup \{i\}$ \tcp*{add bidder $i$ to current solution}
	 
	}
	\Else{
	
		update $A \leftarrow A \setminus \{i\}$ \tcp*{permanently eliminate bidder $i$} \label{algline:endouter}
	}
 }
 }
Let $\bfirstset\leftarrow S_{t-1}$ and $\bsecondset\leftarrow S_t$ \label{algline:prunestart}\\
\If(\tcp*[f]{enforce budget feasibility of $\bfirstset$}){$\sum_{i \in \bfirstset} p_i > B$}{
    let $\lastbidder \leftarrow$ the last bidder added to $S_{t-1}$\\
    update $p_{\lastbidder} \leftarrow \min\{p_{\lastbidder},~ \marginal{\{\lastbidder\}}{S_{t}}\cdot \frac{B}{\gOPT_t}\}$ \label{algline:newprice}\\
    update $\bfirstset \leftarrow \bfirstset \setminus \{\lastbidder\}$\\
    \If{bidder $\lastbidder$  accepts price $p_{\lastbidder}$}{
	 update $\bsecondset \leftarrow \bsecondset \cup \{\lastbidder\}$ \tcp*{move the last bidder $\lastbidder$ to $\bsecondset$} \label{algline:pruneend}
	}
}

\Return{ \textsc{Maximize-Value}($\bfirstset, \bsecondset, p$)} 
 \caption{\mechone, a deterministic  budget-feasible clock auction  for monotone submodular valuation functions}
 \label{alg:4.75apx}
\end{algorithm}

\vspace{3pt}

\renewcommand{\algorithmcfname}{ALGORITHM}

\begin{algorithm}[H]
\setstretch{1.1}
\SetKwInOut{Input}{Input}
\Input{$\bfirstset$, $\bsecondset$ 
and the prices $p_i$ for all $i \in \firstset \cup \bsecondset$ }
let $\secondset \leftarrow$ the longest budget-feasible prefix of $\bsecondset$\\
let $\thirdset \leftarrow$ $\secondset \cup T$, where $T$ is the longest prefix of $\firstset$ such that $\secondset\cup T$ is budget-feasible\\
Let $W \in \{\bfirstset, \thirdset\}$ be the set with the largest value $v(W)$\\
\Return{ $W$ and the corresponding prices}
 \caption{\textsc{Maximize-Value}, an algorithm for maximizing value subject to knapsack constraint}
 \label{alg:Cmerging}
\end{algorithm}
\vspace{5pt}

\renewcommand{\algorithmcfname}{MECHANISM}

After the last phase has concluded, the mechanism lets $\bfirstset$ and $\bsecondset$ denote the sets generated during the last two phases. If $\bfirstset$ is not budget feasible based on the latest prices, the last bidder added to it is removed. That bidder is offered a (weakly) lower price, and if the bidder accepts that lower price it is added to $\bsecondset$. 

At this point, the prices are finalized and what remains is to choose a subset $W$ of active bidders that is budget-feasible (with respect to the final prices), aiming to maximize $v(W)$. Maximizing a submodular function subject to a knapsack constraint is hard to approximate beyond $1-1/e$~\cite{feige1998threshold,sviridenko2004note}, but we show that the \textsc{Maximize-Value} algorithm achieves the desired approximation by just choosing the best out of two simple candidates: i) the set $\firstset$ and ii) the set $\thirdset$. The set $\thirdset$ contains the longest budget-feasible prefix of $\bsecondset$, denoted $\secondset$, and then uses any leftover budget to also hire the longest possible prefix of bidders from $\firstset$ that this leftover budget can buy. Note that, since the prices have been finalized, one can actually replace the call to \textsc{Maximize-Value} with their favorite algorithm for submodular maximization or, even better, just use that algorithm within \textsc{Maximize-Value} to determine another budget-feasible candidate set $W_4$ and just return the set from $\{\firstset,\thirdset, W_4\}$ that gives the highest value for each instance. What is particularly exciting about clock auctions is that one can just plug in any algorithm that they like without affecting the incentives and appealing properties of the auction, which is unlike most other auction formats.

Our main result for this section is that, apart from being a deterministic clock auction, \mechone\ also achieves the best-known approximation for monotone submodular valuations.

\begin{theorem}\label{thm:4.75}
Let $v$ be a monotone submodular valuation function, then
\mechone\ is a polynomial-time deterministic budget-feasible clock auction that achieves a 4.75 approximation.
\end{theorem}

The fact that \mechone\ is budget feasible is easy to verify since the sets considered by \textsc{Maximize-Value} are budget feasible by design, based on the final prices (see Appendix~\ref{subsec:budegtfeasible}). Due to space limitations we also defer the argument that its running time is $O(n^2\log n)$ to Appendix~\ref{subsec:polytimemonotone}, and instead focus on the more challenging argument for proving the approximation factor.

Let $\hat{t}$ denote the last phase before the mechanism terminates and let $\gOPT_t$ denote the target value $\gOPT$ of each phase $t$. Also, let $\mathbf{Q}=\{Q\subseteq \bidders\setminus (\bfirstset\cup \bsecondset)~:~\sum_{i\in Q}c_i \leq B\}$ be the collection of all budget-feasible subsets of bidders in $\bidders\setminus (\bfirstset\cup \bsecondset)$ (the rejected bidders), and let $\R = \argmax_{Q\in \mathbf{Q}}v(Q~|~\bfirstset\cup \bsecondset)$ be the set in $\mathbf{Q}$ that adds the largest marginal value to $\bfirstset \cup \bsecondset$.

To prove the approximation factor, we later argue that $\OPT \leq v(\firstset \cup \bsecondset) + \marginal{R}{\firstset \cup \bsecondset}$, and show that our auction achieves a 4.75-approximation of the benchmark on the right hand side. The following lemma plays a central role in this argument, as it provides an upper bound on the portion of the optimal value lost through rejections, i.e., the second term of the benchmark.

\begin{lemma}\label{lem:main}For any monotone submodular valuation function $v$, if $\hat{t}$ is the last phase of \mechone, we have 
\[\marginal{R}{\firstset \cup \bsecondset} \leq \left(\frac{3}{2}~-~\frac{1}{2^{\hat{t}-2}}\right)\cdot \gOPT_{\hat{t}}.\]
\end{lemma}
\begin{proof}

We denote the last bidder that is added to  $S_{\hat{t}-1}$ (who is offered a new price in line \ref{algline:newprice}) as $\lastbidder$.  We then partition $\R$ into $\R_a$ and $\R_b$ where $\R_a$ consists of the bidders of $\R$ that were rejected in the first $\hat{t}-2$ phases and $\R_b$ consists of bidders in $\R$ rejected in phases $\hat{t} - 1$ and $\hat{t}$ and $\lastbidder$ if $\lastbidder \in \R$. We say bidders in $\R_a$ have total cost $f_a \cdot B$ and those in $\R_b$ have total cost $f_b \cdot B$ where $f_a+f_b\le 1$. 

For any $j \in \R_b$ that is rejected in phase $\hat{t}-1$, let $T_j \subseteq S_{\hat{t}-1}$ be the subset when $j$ is rejected. Note that $T_j$ does not include $\lastbidder$ since $j$ is considered before the addition of $\lastbidder$, so $T_j\subseteq \firstset$, and we have:
\[\marginal{\{j\}}{\firstset \cup \bsecondset} \leq \marginal{\{j\}}{\firstset} \leq \marginal{\{j\}}{T_{j}} \leq \gOPT_{\hat{t}-1} \frac{c_j}{B} \leq \gOPT_{\hat{t}}\frac{c_j}{B}.\]
where the first two inequalities are by submodularity. Similarly, for any $j \in R_b$ that is rejected in phase $\hat{t}$, we have:
\[\marginal{\{j\}}{\firstset \cup \bsecondset} \leq \marginal{\{j\}}{\bsecondset} \leq \marginal{\{j\}}{T_{j}}  \leq \gOPT_{\hat{t}}\frac{c_j}{B}.\]
If the last bidder $\lastbidder$ of $\firstset$ is rejected to make $\firstset$ budget feasible (in line \ref{algline:newprice}), it must be that it rejected the new price the mechanism offers, i.e.,

\[\marginal{\{j^*\}}{\firstset \cup \bsecondset} \leq \marginal{\{j^*\}}{\bsecondset}\leq \gOPT_{\hat{t}}\frac{c_{j^*}}{B}.\]

Since the sum of the costs of the bidders in $R_b$ is equal to $f_b\leq 1$ fraction of budget, i.e., $\sum_{j \in \R_b}c_j =f_bB$. we have:
\[\marginal{\R_b}{\firstset \cup \bsecondset} \leq \sum_{j \in \R_b} \marginal{\{j\}}{\firstset \cup \bsecondset}\leq \sum_{j \in \R_b}\gOPT_{\hat{t}}\frac{c_j}{B} \leq f_b \gOPT_{\hat{t}}. \]

Now let $S^-_t$ denote the longest budget-feasible prefix of $S_t$. We then have $v(S_t^-) \leq \gOPT_t$ for all $t$. We also define $\mathbf{S}_{\hat{t}-2} = \bigcup_{t=2}^{\hat{t}-2}S_t^-$.
For any bidder $j \in \R_a$, by the definition of the mechanism we have $\costi[j] > \marginal{\{j\}}{T_{k}} \cdot  \frac{B}{\gOPT_{k}}$ for some $T_{k} \subseteq S^-_k$ where $2 \leq k \leq \hat{t} - 2$ is the phase where $j$ was rejected. Notice that we don't include $S_1$ since no bidder can be rejected in phase 1. By submodularity, we have  $\marginal{\{j\}}{T_{k}} \geq \marginal{\{j\}}{\mathbf{S}_{\hat{t} - 2}}$. Together with the fact that $\gOPT_{k} \leq \gOPT_{\hat{t}-2}$ for all $k \leq \hat{t}-2$, we get that for all $j \in \R_a$, $\marginal{\{j\}}{\mathbf{S}_{\hat{t}-2}} \leq \gOPT_{\hat{t}-2} \cdot \frac{\costi[i]}{B}.$
Similarly, since the sum of the costs of the bidders in $R_a$ equals to  $f_a$ fraction of the budget, i.e., $\sum_{j \in R_a}c_i = f_aB$, 
\begin{align*}
    \marginal{\R_a}{\mathbf{S}_{\hat{t}-2}} \leq \sum_{j \in \R_a}\marginal{\{j\}}{\mathbf{S}_{\hat{t}-2}}  \leq \sum_{j \in \R_a}\gOPT_{\hat{t}-2} \cdot \frac{\costi[j]}{B} \leq f_a\gOPT_{\hat{t}-2}.
\end{align*}

Recall that $\mathbf{S}_{\hat{t}-2} = \bigcup_{t=2}^{\hat{t}-2}S^-_{t}$, by submodularity we have $v(\mathbf{S}_{\hat{t}-2}) \leq \sum_{t=2}^{\hat{t}-2}v(S^-_t) \leq \sum_{t = 2}^{\hat{t}-2}\gOPT_{t} = \sum_{t=4}^{\hat{t}}\gOPT_{t-2}$ . By monotonicity, 
\begin{align*}
v(\R_a) \leq \marginal{\R_a}{\mathbf{S}_{\hat{t}-2}} + v(\mathbf{S}_{\hat{t}-2}) &\leq f_a\gOPT_{\hat{t}-2} + \sum_{t =4}^{\hat{t}}\gOPT_{t-2} \leq \left(f_a+2-\frac{1}{2^{\hat{t}-4}}\right)\gOPT_{\hat{t}-2} ~~\Rightarrow\\
v(\R_a) &\leq \left(\frac{f_a}{4} +\frac{1}{2}- \frac{1}{2^{\hat{t}-2}}\right)\gOPT_{\hat{t}}
\end{align*}
Combining the analysis of $\R_a$ and $\R_b$, we have:
\begin{align*}
    \marginal{\R}{\firstset \cup \bsecondset} &\leq \marginal{\R_a}{\firstset \cup \bsecondset}+\marginal{\R_b}{\firstset \cup \bsecondset}  \leq v(\R_a) + \marginal{\R_b}{\firstset \cup \bsecondset} ~\Rightarrow\\
    \marginal{\R}{\firstset \cup \bsecondset} & \leq \max_{f_a,f_b~:~f_a+f_b=1}\left(f_b+\frac{f_a}{4}+\frac{1}{2}-\frac{1}{2^{\hat{t}-2}}\right)\gOPT_{\hat{t}} \leq \left(\frac{3}{2}-\frac{1}{2^{\hat{t}-2}}\right) \gOPT_{\hat{t}}.\qedhere
\end{align*}
\end{proof}

Next, we bound the loss in value from potentially discarding the last bidder added to $S_t$ in order to ensure budget feasibility. 

\begin{lemma}
\label{lem:lastelement} Assume that $v$ is a submodular valuation function and
  let  $i_t$ denote the final bidder added to $S_{t}$ in phase $t$. For all $t \geq 2$, if $S_t$ is not budget feasible  $$v(S_t \setminus \{i_t\}) \geq \frac{2^{t-1}}{2^{t-1}+1} \cdot \gOPT_t.$$ 
\end{lemma}
\begin{proof}
Observe that for every bidder $i$ we have $v(\{i\})\leq \gOPT_1$ (by definition of $\gOPT_1$), so $v(\{i\})\leq \frac{\gOPT_{t}}{2^{t-1}}$ for every $t \geq 1$.  Thus, to strictly exceed the target $\gOPT_t$ in any phase $t \geq 2$ we must add at least $2^{t-1}+1$ bidders to the set $S_{t}$.  But since our algorithm considers bidders in weakly decreasing order of marginal contribution, by submodularity we then have that $\marginal{\{i_t\}}{S_t \setminus \{i_t\}} \leq \frac{1}{2^{t-1}+1} \cdot v(S_t)$. Consequently, $v(S_t \setminus \{i_t\}) \geq \frac{2^{t-1}}{2^{t-1}+1} \cdot S_t \geq \frac{2^{t-1}}{2^{t-1}+1} \cdot \gOPT_t$. 
\end{proof}

\begin{proof}[Proof Sketch for Theorem~\ref{thm:4.75}]
\mechone\ is clearly deterministic. Next, note that the sequence of prices offered to a bidder $i$ is descending since at each update of $p_i$, it is the minimum of the previous price $p_i$ and another price. Moreover, once a bidder rejects a price, it exits the auction and is not considered anymore. Thus, \mechone\ is a clock-auction. 

Throughout the proof, we assume $\hat{t}\geq 3$ and $\bsecondset$ is budget-feasible, i.e., $\bsecondset = \secondset$. We show our auction actually achieves a better approximation in the cases where $\hat{t}< 3$ or $\bsecondset$ is not budget feasible 
in Appendices~\ref{subsec:teq2} and \ref{subsec:W2budgetfeasible}, respectively.

Let $\firstset, \secondset, $ and $\thirdset$ denote the sets defined in the \textsc{Maximize-Value} algorithm. We use $\textsc{Benchmark}$ to refer to the value of $v(\firstset\cup \secondset) + \marginal{\R}{\firstset\cup \secondset}$, with the assumption $\secondset = \bsecondset$.  By submodularlity and monotonicity, and since the optimal solution needs to be budget feasible, we have that $\OPT \leq v(\firstset \cup \secondset) + \marginal{R}{\firstset \cup \secondset}.$
 Then, to prove that \mechone\ gives a $\rho$ approximation it is sufficient to show that

\begin{equation*}
\frac{v(\firstset\cup \secondset) + \marginal{\R}{\firstset \cup \secondset}}{\max\{v(\firstset),  v(\thirdset)\}} ~\leq~ \rho. \end{equation*}

Assume, for contradiction, the negation of the above inequality holds true, then it must be that $v(\firstset),v(\thirdset)$ both have value less than $\frac{1}{\rho}$ times \textsc{Benchmark}. We show that for any $\rho \geq 4.75$ this assumption leads to a contradiction. For notational simplicity, we use $\firstconstant$ and $\rconstant$ to denote the constants for which $v(\firstset) = \firstconstant \gOPT_{\hat{t}}$ and $\marginal{\R}{\firstset\cup \secondset} = \rconstant\gOPT_{\hat{t}}$.

\textbullet~ First, from the fact that $v(\firstset)$ is strictly less than $\frac{1}{\rho}$ of the \textsc{Benchmark}, we get
\begin{equation}\label{eq:cons1}
v(\firstset)= \firstconstant\gOPT_{\hat{t}} < \frac{1}{\rho} (v(\combinedset) + \marginal{\R_b}{\combinedset}) ~\Rightarrow~
\frac{v(\combinedset)}{\gOPT_{\hat{t}}} > \left(\rho\firstconstant-\rconstant\right).
\end{equation}

\textbullet~ Then, since $v(\thirdset)$ is strictly less than $\frac{1}{\rho}$ of the \textsc{Benchmark}, and $v(\thirdset)\geq v(\secondset)$, we get

\begin{equation}\label{eq:W2valueBound2}
v(\secondset) \leq v(\thirdset)  <\frac{1}{\rho}\left(v(\combinedset) + \rconstant\gOPT_{\hat{t}}\right) ~\Rightarrow~ v(\combinedset) > \rho v(\secondset)-\rconstant\gOPT_{\hat{t}}.
\end{equation}
The marginal contribution of each bidder $i\in \secondset$ in the order that they were added is at least $\frac{p_i \gOPT_{\hat{t}}}{B}$  so $v(\secondset) \geq \frac{ \gOPT_{\hat{t}}}{B} \sum_{i\in \secondset} p_i$. Thus if we let $u=1-\frac{\sum_{i\in\secondset}p_i}{B}$ be the fraction of the budget left unused by $\secondset$, by Inequality~\eqref{eq:W2valueBound2} we have
\begin{equation}\label{eq:cons2}
v(\combinedset) > \rho (1-u) \gOPT_{\hat{t}} -\rconstant\gOPT_{\hat{t}} ~\Rightarrow~ \frac{v(\combinedset)}{\gOPT_{\hat{t}}} > \rho(1-u) - \rconstant 
\end{equation}

\textbullet~ Furthermore, for the value of $\thirdset$, using submodularity, we get:
\[v(\thirdset)= \marginal{\secondset}{T} +v(T) \geq \marginal{\secondset}{\firstset}+v(T) = v(\firstset \cup \secondset) - \firstconstant\gOPT_{\hat{t}} + v(T)
\]
Using the fact that $v(\thirdset)$ is less than $\frac{1}{\rho}$ of the \textsc{Benchmark} once again, we get
\begin{equation}\label{eq:W3Bound2}
v(\firstset \cup \secondset) -\firstconstant\gOPT_{\hat{t}} + v(T) ~<~
\frac{1}{\rho} (v(\firstset\cup \secondset) + \rconstant\gOPT_{\hat{t}})
\end{equation}
Also, note that for every bidder $i$ we have $v(\{i\})\leq \gOPT_1$ (by definition of $\gOPT_1$), so $v(\{i\})\leq \frac{\gOPT_{t}}{2^{t-1}}$ for every $t \geq 1$. 
Let $T'$ be the shortest prefix of $\firstset$ such that $\sum_{i\in T'}p_i>uB$, i.e., the prefix whose current prices exceed a $u$ fraction of the budget. As each of these bidders was added to $S_{\hat{t}-1}$ in phase $\hat{t}-1$, the ratio of their marginal contribution over the price that they were offered was at least $\frac{\gOPT_{\hat{t}-1}}{B}$, so their total value, $v(T')$ is at least $u\gOPT_{\hat{t}-1}$. If we remove the last bidder from $T'$, we retrieve the set $T$ (the longest prefix of $\firstset$ whose prices add up to at most $uB$ and, hence, can be afforded in addition to $\secondset$). Since that bidder's marginal contribution is at most $\frac{\gOPT_{\hat{t}-1}}{2^{\hat{t}-2}}$ the value of $T$ must be at least
\[v(T) ~\geq~ \left(u - \frac{1}{2^{\hat{t}-2}}\right)\gOPT_{\hat{t}-1} ~=~ \left(u - \frac{1}{2^{\hat{t}-2}}\right)\frac{\gOPT_{\hat{t}}}{2}.\]
Substituting this into~\eqref{eq:W3Bound2} gives 
\begin{align}
\left(1-\frac{1}{\rho}\right)v(\combinedset) &< \frac{\rconstant}{\rho}+\firstconstant-v(T) < \frac{\rconstant}{\rho}+\firstconstant-\left(u-\frac{1}{2^{\hat{t}-2}}\right)\frac{\gOPT_{\hat{t}}}{2} ~\Rightarrow \nonumber\\
\frac{v(\firstset\cup \secondset)}{\gOPT_{\hat{t}}} ~&<~ \frac{2\rho \firstconstant+2 \rconstant-\rho u+\frac{\rho}{2^{\hat{t}-2}}}{2\rho -2} \label{eq:cons3}
\end{align}

In summary, from the assumption that $v(\firstset)$, $v(\secondset)$, and $v(\thirdset)$ all have value less than $\frac{1}{\rho}$ times the \textsc{Benchmark}, we get Inequalities~\eqref{eq:cons1}, \eqref{eq:cons2} and \eqref{eq:cons3}, respectively. The first two inequalities yield lower bounds for the $v(\firstset\cup \secondset)/\gOPT_{\hat{t}}$ ratio, while the third one provides an upper bound. We prove that for any value $\rho\geq 4.75$ these three inequalities are incompatible, leading to a contradiction.

Due to space limitations, we defer the rest of this proof to Appendix~\ref{subsec:thm1}, where we first show that without loss of generality we can assume that $u=1-\firstconstant$. Then, using Lemma~\ref{lem:lastelement} we obtain that $\firstconstant \geq \frac{2^{\hat{t}-2}}{2^{\hat{t}-1}+2}$ and from Lemma~\ref{lem:main} we get that $\rconstant \leq \frac{3}{2}-\frac{1}{2^{\hat{t}-1}}$. Using these two inequalities we verify the incompatibility of Inequalities~\eqref{eq:cons1}, \eqref{eq:cons2} and \eqref{eq:cons3} for $\rho\geq 4.75$, concluding the proof.
\end{proof}

We complement our upper bound of $4.75$ for the approximation factor of \mechone\ with a lower bound of $4.5$. The proof is deferred to Appendix~\ref{subsec:lowerbound}.
\begin{lemma}\label{lem:lowerboundInstance} For any constant $\epsilon > 0$,
there exists a monotone submodular valuation function $v$ for which \mechone \ returns a solution $S$ such that $\OPT > (4.5-\epsilon)v(S)$.
\end{lemma}

\section{Non-Monotone Submodular Valuations}\label{sec:nonmonotone}

In this section, we develop a deterministic clock auction that achieves a constant factor approximation for submodular valuation functions (not necessarily monotone) and runs in polynomial time. This is the first deterministic budget-feasible mechanism for general submodular valuation functions that achieves a constant factor approximation, even for non-polynomial time mechanisms, and we achieve this with a clock auction that runs in polynomial time. The mechanism combines the backward greedy technique from the previous section and techniques for maximizing non-monotone submodular  functions.

Similarly to \mechone \ from the previous section, \mechtwo, formally described below as Mechanism~\ref{alg:nonmonotone}, proceeds in phases and aims to find a set of bidders with value at least $\gOPT_t$ at each phase $t$. The main difference with \mechone \ is that, instead of constructing a single tentative set $S_t$ of bidders at each phase $t$,  \mechtwo \ constructs two disjoint tentative sets $S^1_t$ and $S^2_t$ of bidders at each phase. This technique of constructing two disjoint sets of bidders to handle  non-monotone valuation functions in budget-feasible mechanism design was introduced  by \citet{amanatidis2019budget} with a mechanism called \textsc{Simultaneous Greedy}. \mechtwo \ integrates this technique in the \mechone \ mechanism designed for monotone valuation functions.

At each iteration of phase $t$, the mechanism considers bidder $i$ and set of bidders $S^k_t \in \{S^1_t, S^2_t\}$ such that the marginal contribution $\marginal{\{i\}}{S^k_t}$ of $i$ to $S^k_t$ is maximized. It then offers price $p_i = \min \left\{p_i,~ \marginal{\{i\}}{S^k_{t}} \cdot \frac{B}{\gOPT}\right\}$ to bidder $i$, adds bidder $i$ to $S_t^k$ if $i$ accepts price $p_i$, and permanently eliminates bidder $i$ otherwise. A phase terminates when either $S^1_t$ or $S^2_t$ reaches the target $\gOPT_t$, or when there are no more bidders to offer a price to. At the beginning of a new phase $t$, the mechanism sets aside both $S^1_{t-1}$ and $S^2_{t-1}$.

\vspace{.3cm}

\begin{algorithm}[H]\label{alg:nonmonotone}
\setstretch{1.25}
\SetKwInOut{Input}{Input}
\Input{Budget $B$,   valuation function $v: 2^N \rightarrow \mathbb{R}$}
 initialize $A \leftarrow \bidders$, $S^1_{0}, S^2_{0}, S^1_{1} \leftarrow \emptyset, S^2_1 \leftarrow \left\{\text{argmax}_{i \in \bidders} \ \vali\right\}$,  $\gOPT \leftarrow v(S^2_1)$, $t \leftarrow 1$,  $p_i \leftarrow B$ for all $i \in \bidders$

\While{$A \setminus (S^1_{t-1} \cup S^2_{t-1} \cup S^1_{t} \cup S^2_{t})\neq \emptyset$}{

update $t \leftarrow t+1$, $\gOPT \leftarrow 2 \gOPT$ and initialize $S^1_t, S^2_t \leftarrow \emptyset$
 \tcp*{Start  a new phase}
 
 \While{$\max \{v(S^1_{t}), v(S^2_{t})\} < \gOPT$ and $A \setminus (S^1_{t-1} \cup S^2_{t-1} \cup S^1_{t} \cup S^2_{t}) \neq \emptyset$}{
 
  let $(i, k) \leftarrow \text{argmax}_{i \in A \setminus \left(S^1_{t-1} \cup S^2_{t-1} \cup S^1_{t} \cup S^2_{t}\right), k \in [2]}\marginal{i}{ S^k_{t}}$ 
  
  update $p_i \leftarrow \min \left\{p_i, \marginal{\{i\}}{S^k_{t}} \cdot \frac{B}{\gOPT}\right\}$
  
	\If{bidder $i$ accepts price $p_i$}{
	 $S^k_{t} \leftarrow S^k_{t} \cup \{i\}$ \tcp*{Add bidder $i$ to current solution}
	 
	}
	\Else{
	
		$A \leftarrow A \setminus \{i\}$ \tcp*{Permanently discard bidder $i$}
	}
 }
 }
 let $T^k_{j} \leftarrow \unmax(v, S^k_{j})$, for $j \in \{t-1, t\}$ and $k \in \{1,2\}$

let $\Sfinal \leftarrow \argmax_{S' \in \{ S^1_{t-1}, S^2_{t-1}, 
T^1_{t}, T^2_{t}, T^1_{t-1}, T^2_{t-1}\}} v(S')$

\If(\tcp*[f]{ensure budget feasibility}){$\sum_{i \in \Sfinal} p_i > B$ }{
    update $\Sfinal \leftarrow \Sfinal \setminus \{i\}$ where $i$ is the last bidder  added to $\Sfinal$ 
}
\Return{$\Sfinal$  and prices $p_i$ for each bidder $i \in \Sfinal$}

 \caption{\textsc{Simultaneous-Iterative-Pruning}, a deterministic budget-feasible clock auction for non-monotone submodular valuation functions}
\end{algorithm}

\vspace{.3cm}
After the last phase $t$, the mechanism runs an unconstrained submodular maximization algorithm that achieves a $2$-approximation, for example the algorithm of \citet{buchbinder2015tight}, over valuation function $v$ and ground set of bidders $S_j^k$ to obtain sets $T_j^k$ such that $T_j^k \geq \frac{1}{2}\max_{T \subseteq S_j^k} v(T)$ for each set of bidder $S_j^k$ constructed in one of the last two phases of the mechanism. Finally, we return the set $\Sfinal$ of bidders of maximum value among  $6$ solutions constructed during the last two phases, but without the last bidder added to $\Sfinal$ if $\Sfinal$ is not budget feasible.

Our main result for this section is that, apart from being a deterministic clock auction,\\ \mechtwo \ also achieves the best-known approximation for non-monotone submodular valuations.

\begin{theorem} 
\label{thm:nmmain}
Let $v$ be a submodular valuation function, then
\mechtwo \ is a polynomial-time deterministic budget-feasible clock auction that achieves a $64$-approximation.
\end{theorem}

The proof that it is budget feasible is identical to the proof of Lemma~\ref{lem:budgetfeasible} which shows the budget feasibility of \mechone. For the running time, the proof that the outer-while loop terminates in polynomial time is identical to the proof that \mechone \ is a polynomial time mechanism. Finally, for the \unmax \ subroutine, we use a $2$-approximation algorithm for unconstrained non-monotone submodular maximization, such as the algorithm by \citet{buchbinder2015tight} which is polynomial time. Thus, \mechtwo \ is a polynomial-time mechanism.

We now turn toward showing that the mechanism achieves a $64$-approximation. Lemma~\ref{lem:NM-rejectedEachRound} bounds the loss from optimal bidders who were eliminated in each phase $t$. However, the proof for this bound on the loss from eliminated optimal bidders is different from the proof from the previous section which assumes monotonicity. To handle non-monotone valuation functions, the proof exploits the fact that we construct two sets $S^1_t$ and $S^2_t$ at each phase. For non-monotone valuation functions, it is also not sufficient to lower bound the value of active optimal bidders, which can be larger than the value of all active bidders. Next, Lemma~\ref{thm:nmbeforepruning} uses the sets $T_j^k$ obtained by running an unconstrained non-monotone submodular maximziation algorithm to approximate the value of active optimal bidders.

We begin by bounding  the loss from optimal bidders who were eliminated in each round $t$.

\begin{lemma}\label{lem:NM-rejectedEachRound} Assume that $v$ is a submodular valuation function and
let $O^-_t$ denote the subset of optimal bidders rejected in phase $t$.  Then, for all $t \leq \hat{t}$, we have that \[v(O^-_t) \geq 6\gOPT_t.\]
\end{lemma}
\begin{proof}
 By submodularity and non-negativity, we know that \[v(O^-_t) \leq v(O^-_t) + v(O^-_t \cup S^1_t \cup S^2_t) \leq v(O_t^- \cup S_t^1) + v(O_t^- \cup S_t^2).\]  We can bound the terms $v(O^-_t \cup S_t^j)$ for $j \in \{1,2\}$ separately.  By submodularity, we know that 
\begin{align*}
    v(O^-_t \cup S_t^1) \leq v(S_t^1) + \sum_{i \in O^-_t}{\marginal{i}{S_t^1}}.
\end{align*}

On the other hand, each $i \in O^-_t$ was rejected because when it was offered a new price, this price was too low.  Let $S^{1,i}_t$ denote the set $S^1_t$ at the point when $i$ was offered a new price.  Then we have that $\marginal{i}{S^{1,i}_t} \geq \marginal{i}{S^1_t}$.  Thus, we have
\begin{align*}
    v(O^-_t \cup S_t^1) \leq v(S_t^1) + \sum_{i \in O^-_t}{c_i \cdot \frac{\gOPT_t}{B}} \leq v(S_t^1) + \gOPT_t.
\end{align*}
Similarly, \[v(O^-_t \cup S_t^2) \leq v(S_t^2) + \gOPT_t.\]

Since for every phase $t \geq 1$ we know $v(\{i\}) \leq \gOPT_t$ we have that $v(S_t^j) \leq 2\gOPT_t$ for $j \in \{1,2\}$.  Combining these inequalities, we have that $v(O_t^-) \leq 6\gOPT_t$ for all $t$, completing the proof.
\end{proof}

By consequence of Lemma \ref{lem:NM-rejectedEachRound}, if we let $O^+_t$ denote the set of bidders in the optimal solution which \emph{remain active} at the \emph{end} of phase $t$ we obtain the following corollary
\begin{corollary}\label{cor:NM-remainingValue} Assume that $v$ is a submodular valuation function and
let $O^+_t$ denote the set of bidders in the optimal solution which are not rejected by the end of phase $t$, we then have
\[v(O^+_t) \geq \OPT - 12\gOPT_t.\]
\end{corollary}
\begin{proof}
By submodularity and Lemma \ref{lem:NM-rejectedEachRound} we have
 \[v(O^+_t) \geq v(O) - \sum_{t' \in [t]}{v(O^-_{t'})} \geq \OPT - \sum_{t' \in [t]}{6\gOPT_{t'}}.\]
Since the value of $\gOPT$ increases by a factor of two in each phase, we can rewrite our above bound as $v(O^+_t) \geq \OPT - 12\gOPT_t$, completing the proof.
\end{proof}

With Corollary \ref{cor:NM-remainingValue} in hand, we can now give a bound on the approximation obtained from the best set in $\{ S^1_{\hat{t}-1}, S^2_{\hat{t}-1}, 
T^1_{\hat{t}}, T^2_{\hat{t}}, T^1_{\hat{t}-1}, T^2_{\hat{t}-1}\}$.

\begin{lemma}\label{thm:nmbeforepruning}
For any submodular valuation function $v$, we have $$\max_{S' \in \{ S^1_{\hat{t}-1}, S^2_{\hat{t}-1}, 
T^1_{\hat{t}}, T^2_{\hat{t}}, T^1_{\hat{t}-1}, T^2_{\hat{t}-1}\}} v(S') \geq \frac{\OPT}{32}$$
where $\hat{t}$ is the last phase of the mechanism.
\end{lemma}
\begin{proof}
There are two cases based on when the last phase $\hat{t}$ of the mechanism occurs. First, if $\gOPT_{\hat{t}} \geq \frac{\OPT}{16}$, then we have $\max\{v(S^1_{\hat{t}-1}), v(S^2_{\hat{t}-1})\} \geq  \frac{\gOPT_{\hat{t}}}{2} \geq \frac{\OPT}{32}.$

Otherwise, $\gOPT_{\hat{t}} \leq \frac{\OPT}{16}$, which is the main case. Since all of the bidders in $O^+_{\hat{t}}$ remain active at the end of phase $\hat{t}$, we know that any bidder $i \in O^+_{\hat{t}}$ must be contained in one of our four candidate solutions: $S^1_{\hat{t}-1}$, $S^2_{\hat{t}-1}$, $S^1_{\hat{t}}$, $S^2_{\hat{t}}$.  But then, by submodularity we have that
\begin{equation}\label{eq:coverOPlus}
    v(O^+_{\hat{t}} \cap S^1_{\hat{t}-1}) + v(O^+_{\hat{t}} \cap S^2_{\hat{t}-1}) + 
    v(O^+_{\hat{t}} \cap S^1_{\hat{t}}) +
    v(O^+_{\hat{t}} \cap S^2_{\hat{t}}) \geq v(O^+_{\hat{t}})
\end{equation}

After our auction completes offering new prices to all active bidders (i.e., the outer while loop terminates), we apply an unconstrained optimizer.
After applying the unconstrained $2$-approximate non-monotone submodular function maximizer, e.g. from \citet{buchbinder2015tight}, to each of our four candidate solutions to obtain $T^k_{\hat{t}-1}$ and $T^k_{\hat{t}}$ for $k \in [2]$.  We then have
\begin{equation}\label{eq:fourSolutions}
    2v(T^1_{\hat{t}-1})+2v(T^2_{\hat{t}-1})+2v(T^1_{\hat{t}})+2v(T^2_{\hat{t}}) \geq v(O^+_{\hat{t}}).
\end{equation}
Let $T^* = \text{argmax}_{T \in \{T^1_{\hat{t}-1}, T^2_{\hat{t}-1}, T^1_{\hat{t}}, T^2_{\hat{t}}\}}{v(T)}$.  From Equation \eqref{eq:fourSolutions} we can then observe
\begin{equation}\label{eq:8covers}
8v(T^*) \geq v(O^+_{\hat{t}}).
\end{equation}
Next, combining Equation \eqref{eq:8covers} and Corollary \ref{cor:NM-remainingValue} gives
\begin{equation}\label{eq:mainEquation}
v(T^*) \geq \frac{\OPT}{8} - \frac{3}{2}\gOPT_{\hat{t}} \geq \frac{\OPT}{32}
\end{equation}
where the last inequality is since  $\gOPT_{\hat{t}} \leq \frac{\OPT}{16}$.
\end{proof}

We are now ready to complete the proof of Theorem~\ref{thm:nmmain} by showing that \mechtwo\ achieves a $64$-approximation to the optimal value.

\begin{proof}[Proof for Theorem~\ref{thm:nmmain}]
\mechtwo \ is clearly deterministic. Next, note that the sequence of prices offered to a bidder $i$ is descending since each update of $p_i$ is the minimum of the previous price and another price. Moreover, once a bidder rejects a price, it exits the auction and is not considered anymore. Thus, \mechtwo \ is a clock-auction.

There are two  cases based on $\Sfinal$. If $\Sfinal$ is initialized to $S_j^k$ or $T_j^k$ with $j = 1$, then we have $\Sfinal = \left\{\text{argmax}_{i \in \bidders} \ \vali\right\}$ and since  $\sum_{i \in \Sfinal} p_i = B$, the mechanism does not remove a bidder from $\Sfinal$ and 
we have 
$v(\Sfinal) \geq v(T^*) \geq   \frac{\OPT}{32}$
by Lemma~\ref{thm:nmbeforepruning}.

 Otherwise, $j > 1$, and  note that by definition of $\gOPT_1$ and by submodularity, for any set $T$ and bidder $i$, we have $\marginal{i}{T} \leq \max_{i'} v(i') = \gOPT_1 \leq \frac{\gOPT_t}{2} \leq \frac{\max\{v(S_t^1), v(S_t^2)\}}{2}$ for all $t \geq 2$. Let $i$ be the potential bidder that was removed from $\Sfinal$ before $\Sfinal$ is returned. Then,
$$v(\Sfinal) \geq v(\Sfinal \cup \{i\}) - \frac{1}{2} \cdot \max\{v(S_t^1), v(S_t^2)\} \geq \frac{1}{2}\max_{S' \in \{ S^1_{\hat{t}-1}, S^2_{\hat{t}-1}, 
T^1_{\hat{t}}, T^2_{\hat{t}}, T^1_{\hat{t}-1}, T^2_{\hat{t}-1}\}} v(S') \geq \frac{\OPT}{64}$$
where the second inequality is since $\Sfinal \cup \{i\} = \argmax_{S' \in \{ S^1_{\hat{t}-1}, S^2_{\hat{t}-1}, 
T^1_{\hat{t}}, T^2_{\hat{t}}, T^1_{\hat{t}-1}, T^2_{\hat{t}-1}\}} v(S')$ and the last inequality is by Lemma~\ref{thm:nmbeforepruning}.
\end{proof}

\section{Subadditive Valuations}\label{sec:subadditive}
\newcommand{\prevset}{S_{\text{prev}}}
\newcommand{\currset}{S_{\text{curr}}}
\newcommand{\finalset}{S_{\text{final}}}
\newcommand{\opt}{O}
In this section we present a secondary result demonstrating how our method of gradually refining an estimate of \OPT\ while maintaining value monotonicity can be used to derandomize  the (randomized) budget-feasible auction of \citet{bei2017worst}. This auction achieves the best known approximation of $O(\log{n}/\log{\log{n}})$ for subadditive valuations, and our deterministic auction matches this bound. The resulting deterministic auction improves upon the previous best deterministic auction of \citet{dobzinski2011mechanisms} which achieves only a $O(\log^3{n})$-approximation. We note that, unlike our auctions for submodular valuation functions that use value queries, the following auction uses demand queries. This is due to the fact that, as we highlighted above, no non-trivial approximation can be achieved using a polynomial number of value queries when maximizing an XOS function \citep[Theorem 6.2]{amanatidis2019budget} and all XOS functions are subadditive.

\vspace{5pt}
\begin{algorithm}[H] \label{alg:subadditive}
	\SetKwInOut{Input}{Input}
	\Input{Budget $B$,   valuation function $v: 2^N \rightarrow \mathbb{R}$}
	initialize  $A \leftarrow \mathcal{N}$, $\prevset \leftarrow \emptyset, \currset \leftarrow \emptyset$, $t \leftarrow 0$\;
	
	\While{$A \setminus \left(\prevset \cup \currset \right) \neq \emptyset$}{		
		$t \leftarrow t + 1$		\tcp*{start  a new phase}
		
		$\prevset \leftarrow \text{argmax}_{S \in \{\prevset, \currset\}}\{v(S)\}$\;
		
		update price $p_i$ of each bidder $i \in A \setminus \prevset$ to $B/t$\;
			
		\If{bidder $i$  rejects new price}{
			update $A \leftarrow A \setminus \{i\}$ \tcp*{permanently eliminate bidder $i$}
		}
	
		Let $\currset$ be the feasible subset of $A \setminus \prevset$ returned by the $2$-approximation algorithm of \citet{badanidiyuru2012optimization}  at the current price level\;
	}

	let $\finalset \leftarrow \text{argmax}_{S \in \{\prevset, \currset\}}\{v(S)\}$\;
	
	\Return{$\finalset$  and prices $p_i$ for each bidder $i \in \finalset$}
	
	 \caption{A deterministic budget-feasible clock auction for subadditive valuation functions}
\end{algorithm}
\vspace{5pt}

Notice that our auction for subadditive valuations follows a similar template to our auctions for submodular valuations.  While for submodular valuations we gradually increase a benchmark value for the subset of bidders we include in our temporary solution, in Mechanism \ref{alg:subadditive} we gradually increase a benchmark \emph{size} for the subset of bidders we include in our temporary solution.  Notably, in both settings we maintain two solutions which ensures that the value that our auction obtains is monotone non-decreasing.  This is a key tool that allows for the approximation guarantees of all of our auctions.  We note that the analysis of the approximation factor of our auction follows almost directly from the analysis in \citet{bei2017worst}.  We include a detailed proof of the approximation factor for completeness, below.

\begin{theorem}
Let $v$ be a subadditive valuation function, then Mechanism \ref{alg:subadditive} is a polynomial-time deterministic budget-feasible clock auction that obtains a $O(\log{n}/\log{\log{n}})$-approximation. 
\end{theorem}

\begin{proof}
Let $\opt = \{1,2,3,\dots,m\}$ denote the optimal set of bidders indexed in non-increasing cost order, i.e., $c_1 \geq c_2 \geq \dots c_m$.  We divide the agents from $\opt$ into disjoint subsets $Z_1, \dots, Z_{r+1}$ such that $Z_1$ contains the $\floor{\frac{B}{c_1}}$ first bidders.  For all $i \geq 2$, let $j(i)$ denote the bidder in $\opt$ of largest cost not contained in any $Z_k$ for all $k < i$.  We then may define $Z_i$ as the $\left\lfloor{\frac{B}{c_{j(i)}}}\right\rfloor$ first bidders beginning at $j(i)$ (or fewer if we exhaust all bidders).

We now proceed via case-analysis on the sets $\{Z_i\}_{i \in [r+1]}$.  First suppose that there exists some set $Z_i$ with $v(Z_i) \geq \frac{\log{\log{n}}}{10\log{n}} \cdot v(\opt)$.  We argue that the mechanism then outputs a set of value at least $\frac{\log{\log{n}}}{40\log{n}} \cdot v(\opt)$. By definition, each bidder in $Z_i$ has cost less than or equal to $c_{j(i)}$, and note that we offer price $B/k$ to all bidders in round $k$. But then we must offer prices weakly above the cost of all bidders in $Z_i$ for all rounds up to and including $|Z_i|$. Thus all of the bidders in $Z_i$ will be active at the point when the price of $\frac{B}{|Z_i|}$ is offered to all bidders in the auction. Hence, if no bidders in $Z_i$ are included in $\prevset$ we will identify a set of value at least $\frac{1}{2}v(Z_i)$ in this phase of the auction.  Since our auction obtains value equal to the set of highest value identified in any phase, we are done.  Suppose not, that is, suppose that some portion of $Z_i$ is contained in $\prevset$.  Since, in each phase, we identify some feasible set giving a $2$-approximation to the highest achievable value given the current prices, we know that $v(\prevset \cap Z_i) \leq 2v(\prevset)$.  Moreover, we know that the set $\currset$ that we select in phase $|Z_i|$ is such that $v(Z_i \setminus \prevset) \leq 2v(\currset)$.  But then, by subadditivity, we have that
	\begin{equation*}
		v(Z_i) \leq v(\prevset \cap Z_i) + v(Z_i \setminus \prevset)
		\leq 2v(\prevset) + 2v(\currset)
		\leq 4\cdot \text{max}\left\{v(\prevset), v(\currset)\right\}.
	\end{equation*}
Finally, since our auction obtains value equal to the highest value identified in any phase, we know that we obtain a set of value $\frac{\log{\log{n}}}{40\log{n}}\cdot v(\opt)$.  Thus, we obtain a $O\left(\frac{\log{n}}{\log{\log{n}}}\right)$-approximation whenever there exists some set $Z_i$ with $v(Z_i) \geq \frac{\log{\log{n}}}{10\log{n}} \cdot v(\opt)$.

We now deal with the other case.  That is, suppose that for all $i \in [r+1]$ we have that $v(Z_i) < \frac{\log{\log{n}}}{10\log{n}}\cdot v(\opt)$.  By subadditivity, we know that
\[\sum_{i = 1}^{r+1}{v(Z_i) \geq v(\opt).}\]
But then, it must be that $(r+1)\cdot\frac{\log{\log{n}}}{10\log{n}}\cdot v(\opt) > v(\opt)$, which implies that
\begin{equation}\label{eq:sizeLB}
	r > \frac{10\log{n}}{\log{\log{n}}} - 1 \geq \frac{5\log{n}}{\log{\log{n}}} \geq \frac{5\log{m}}{\log{\log{m}}}.
\end{equation}

We know that $\opt$ is budget feasible, i.e., $\sum_{i \in m}{c_i} \leq B$.  Also, by construction we have $c_{j(i)}>\frac{B}{|Z_i| + 1}$ for all $i \in [r]$.  Combining these gives
\begin{align*}
	B &\geq \sum_{j = 1}^{m}{c_j}\\
	&\geq c_1 + |Z_1| \cdot c_{j(2)} + |Z_2|\cdot c_{j(3)} + \cdots + |Z_r|\cdot c_{j(r+1)}\\
	&>\frac{B}{|Z_1| + 1} + \frac{|Z_1| \cdot B}{|Z_2| + 1} + \cdots + \frac{|Z_{r-1}| \cdot B}{|Z_{r}| + 1}.
\end{align*}
Note that the only possibly empty set is $Z_{r+1}$ by construction.  Thus, for all $i < r+1$ we know that $2|Z_i| \geq |Z_i| + 1$.  We then may conclude that
\begin{align*}
	1 &\geq \frac{1}{|Z_1| + 1} + \frac{|Z_1|}{|Z_2| + 1} + \cdots + \frac{|Z_{r-1}|}{|Z_{r}| + 1} \\
	&\geq \frac{1}{2|Z_1|} + \frac{|Z_1|}{2|Z_2|} + \cdots + \frac{|Z_{r-1}|}{2|Z_{r}|} \\
	&\geq \frac{1}{2} \cdot r \left[\frac{1}{|Z_1|} \cdot \frac{|Z_1|}{|Z_2|}\cdot \cdots \cdot \frac{|Z_{r-1}|}{|Z_{r}|}\right]^{1/r},
\end{align*} 
where the last step uses the AM-GM inequality.
Simplifying gives $2 \geq r \cdot \left(\frac{1}{|Z_r|}\right)^{1/r}$, i.e., $|Z_r| \geq \left(\frac{r}{2}\right)^r$.  On the other hand, we have that $m \geq |Z_r|$.  Combining these two with Equation \eqref{eq:sizeLB} we then have
\begin{align*}
	\log{m} &\geq r \cdot \log{\frac{r}{2}}\\
	&\geq \frac{5 \log{m}}{\log{\log{m}}} \cdot \left(\log{\log{m}} - \log{\log{\log{m}}} + \log{\frac{5}{2}}\right)
\end{align*}
which is a contradiction.
In other words, it must be that there exists some $Z_i$ with $v(Z_i) \geq \frac{\log{\log{n}}}{10\log{n}}\cdot v(\opt)$, completing the proof.
\end{proof}

\section{Conclusion}
With the auctions that we propose in this paper, we significantly improve our understanding of budget-feasible mechanism design in two important ways: 

First, our auctions achieve improved approximation factors in a \emph{deterministic} fashion and resolve one of the main open problems in the area. In contrast to some prior work that depends on randomized sampling in order to estimate the optimal value, we instead introduce a deterministic discovery process with a primal-dual flavor. We start with a low estimate of the optimal value, which we use in order to determine the initial prices offered to the bidders. Then, depending on the bidders' responses to these prices (i.e., depending on which bidders accept the prices offered to them), we update our estimate and repeat this process. This way, our auction gradually refines its estimate of the optimal value, while simultaneously discovering the appropriate prices for approximating the optimal value in a budget-feasible way. 

Second, our solutions takes the form of a \emph{clock auctions}. Unlike sealed-bid auctions, where the bidders directly report their costs to the auctioneer, clock auctions can only assess these costs indirectly, by offering a sequence of descending prices to the bidders. The price discovery process described above meets this restriction, and gradually develops a better understanding of the bidders' true costs. The fact that our solutions are clock auctions implies that they satisfies a list of highly desirable properties, making them more attractive for practical applications. Another implication, which is particularly interesting from a theoretical perspective, is that they yield non-trivial backward greedy algorithms for submodular maximization, which nicely complement the existing literature on submodular maximization, which is dominated by forward greedy algorithms.

\paragraph{Limitations of posted-price mechanisms}
To complement our positive results regarding the ability of budget-feasible clock auctions to achieve a constant factor approximation, we also considered the special class of clock auctions that take the form of posted-price mechanisms. These mechanisms approach the bidders in some order and offer them a take-it-or-leave-it price. This approach proved useful for the design of randomized clock auctions that can use sampling to estimate the optimal value~\cite{bei2017worst,amanatidis2019budget}. We were able to verify that, without the estimate that the randomized sampling provides, these mechanisms are insufficient for achieving any non-trivial approximation, suggesting that the approach of~\citet{bei2017worst} and \citet{amanatidis2019budget} could not be extended toward a deterministic solution. Due to space constraints, these results have been deferred to Appendix~\ref{sec:postedprice}.

\paragraph{Future directions}
Our results provide an optimistic view toward the design of practical budget-feasible auctions, and they give rise to interesting open problems such as the following:
\begin{itemize}
    \item Is there a separation between the performance of the best possible budget-feasible clock auction and the best possible strategyproof budget-feasible mechanism?
    \item Do there exist budget-feasible clock auctions that can achieve a constant factor approximation beyond submodular valuations (e.g., for subadditive valuations)?
\end{itemize}

Regarding the first question, there is no known separation between clock auctions and general strategyproof mechanisms, even for interesting special classes of valuations, such as additive ones. Note that the best known approximation guarantees (for both randomized and deterministic auctions), for the case of additive valuations, are currently due to \citet{gravin2019optimal}. Although these auctions are presented as sealed-bid mechanisms, we were able to verify that they are one of the few examples that can also be implemented as clock auctions. As a result, for the special case of additive valuations, the state of the art approximations can be achieved by clock auctions.

Regarding the second question, for the more general class of strategyproof mechanisms, we know that there exists a constant factor mechanism, through a non-constructive argument, based on Yao's lemma due to~\citet{bei2017worst}. Therefore, designing a specific strategyproof auction that achieves this guarantee remains open. However, focusing our attention on the more restrictive class of clock auctions can help us gain some traction on this problem. For example, this restriction would make it more tractable to prove larger lower bounds; something that would have been much more demanding for the richer class of strategyproof mechanisms.

\newpage

\begin{appendices}
\section{Proofs missing from Section~\ref{sec:monotone}}
\label{sec:appendixmonotone}

\subsection{Proof of Theorem~\ref{thm:4.75}}\label{subsec:thm1}
\begin{proof}
\mechone\ is clearly deterministic. Next, note that the sequence of prices offered to a bidder $i$ is descending since at each update of $p_i$, it is the minimum of the previous price $p_i$ and another price. Moreover, once a bidder rejects a price, it exits the auction and is not considered anymore. Thus, \mechone\ is a clock-auction. 

Throughout the proof, we assume $\hat{t}\geq 3$ and $\bsecondset$ is budget-feasible, i.e., $\bsecondset = \secondset$. We show our auction actually achieves a better approximation in the cases where $\hat{t}< 3$ or $\bsecondset$ is not budget feasible 
in Appendices~\ref{subsec:teq2} and \ref{subsec:W2budgetfeasible}, respectively.

Let $\firstset, \secondset, $ and $\thirdset$ denote the sets defined in the \textsc{Maximize-Value} algorithm. We use $\textsc{Benchmark}$ to refer to the value of $v(\firstset\cup \secondset) + \marginal{\R}{\firstset\cup \secondset}$, with the assumption $\secondset = \bsecondset$.  By submodularlity and monotonicity, and since the optimal solution needs to be budget feasible, we have that $\OPT \leq v(\firstset \cup \secondset) + \marginal{R}{\firstset \cup \secondset}.$
 Then, to prove that \mechone\ gives a $\rho$ approximation it is sufficient to show that

\begin{equation*}
\frac{v(\firstset\cup \secondset) + \marginal{\R}{\firstset \cup \secondset}}{\max\{v(\firstset),  v(\thirdset)\}} ~\leq~ \rho. \end{equation*}

Assume, for contradiction, the negation of the above inequality holds true, then it must be that $v(\firstset),v(\thirdset)$ both have value less than $\frac{1}{\rho}$ times \textsc{Benchmark}. We show that for any $\rho \geq 4.75$ this assumption leads to a contradiction. For notational simplicity, we use $\firstconstant$ and $\rconstant$ to denote the constants for which $v(\firstset) = \firstconstant \gOPT_{\hat{t}}$ and $\marginal{\R}{\firstset\cup \secondset} = \rconstant\gOPT_{\hat{t}}$.

\textbullet~ First, from the fact that $v(\firstset)$ is strictly less than $\frac{1}{\rho}$ of the \textsc{Benchmark}, we get
\begin{equation}
v(\firstset)= \firstconstant\gOPT_{\hat{t}} < \frac{1}{\rho} (v(\combinedset) + \marginal{\R_b}{\combinedset}) ~\Rightarrow~
\frac{v(\combinedset)}{\gOPT_{\hat{t}}} > \left(\rho\firstconstant-\rconstant\right).
\end{equation}

\textbullet~ Then, since $v(\thirdset)$ is strictly less than $\frac{1}{\rho}$ of the \textsc{Benchmark}, and $v(\thirdset)\geq v(\secondset)$, we get

\begin{equation}\label{eq:W2valueBound21}
v(\secondset) \leq v(\thirdset)  <\frac{1}{\rho}\left(v(\combinedset) + \rconstant\gOPT_{\hat{t}}\right) ~\Rightarrow~ v(\combinedset) > \rho v(\secondset)-\rconstant\gOPT_{\hat{t}}.
\end{equation}
The marginal contribution of each bidder $i\in \secondset$ in the order that they were added is at least $\frac{p_i \gOPT_{\hat{t}}}{B}$  so $v(\secondset) \geq \frac{ \gOPT_{\hat{t}}}{B} \sum_{i\in \secondset} p_i$. Thus if we let $u=1-\frac{\sum_{i\in\secondset}p_i}{B}$ be the fraction of the budget left unused by $\secondset$, by Inequality~\eqref{eq:W2valueBound21} we have
\begin{equation}
v(\combinedset) > \rho (1-u) \gOPT_{\hat{t}} -\rconstant\gOPT_{\hat{t}} ~\Rightarrow~ \frac{v(\combinedset)}{\gOPT_{\hat{t}}} > \rho(1-u) - \rconstant 
\end{equation}

\textbullet~ Furthermore, for the value of $\thirdset$, using submodularity, we get:
\[v(\thirdset)= \marginal{\secondset}{T} +v(T) \geq \marginal{\secondset}{\firstset}+v(T) = v(\firstset \cup \secondset) - \firstconstant\gOPT_{\hat{t}} + v(T)
\]
Using the fact that $v(\thirdset)$ is less than $\frac{1}{\rho}$ of the \textsc{Benchmark} once again, we get
\begin{equation}\label{eq:W3Bound2a1}
v(\firstset \cup \secondset) -\firstconstant\gOPT_{\hat{t}} + v(T) ~<~
\frac{1}{\rho} (v(\firstset\cup \secondset) + \rconstant\gOPT_{\hat{t}})
\end{equation}
Also, note that for every bidder $i$ we have $v(\{i\})\leq \gOPT_1$ (by definition of $\gOPT_1$), so $v(\{i\})\leq \frac{\gOPT_{t}}{2^{t-1}}$ for every $t \geq 1$. 
Let $T'$ be the shortest prefix of $\firstset$ such that $\sum_{i\in T'}p_i>uB$, i.e., the prefix whose current prices exceed a $u$ fraction of the budget. As each of these bidders was added to $S_{\hat{t}-1}$ in phase $\hat{t}-1$, the ratio of their marginal contribution over the price that they were offered was at least $\frac{\gOPT_{\hat{t}-1}}{B}$, so their total value, $v(T')$ is at least $u\gOPT_{\hat{t}-1}$. If we remove the last bidder from $T'$, we retrieve the set $T$ (the longest prefix of $\firstset$ whose prices add up to at most $uB$ and, hence, can be afforded in addition to $\secondset$). Since that bidder's marginal contribution is at most $\frac{\gOPT_{\hat{t}-1}}{2^{\hat{t}-2}}$ the value of $T$ must be at least
\[v(T) ~\geq~ \left(u - \frac{1}{2^{\hat{t}-2}}\right)\gOPT_{\hat{t}-1} ~=~ \left(u - \frac{1}{2^{\hat{t}-2}}\right)\frac{\gOPT_{\hat{t}}}{2}.\]
Substituting this into~\eqref{eq:W3Bound2a1} gives 
\begin{align}
\left(1-\frac{1}{\rho}\right)v(\combinedset) &< \frac{\rconstant}{\rho}+\firstconstant-v(T) < \frac{\rconstant}{\rho}+\firstconstant-\left(u-\frac{1}{2^{\hat{t}-2}}\right)\frac{\gOPT_{\hat{t}}}{2} ~\Rightarrow \nonumber\\
\frac{v(\firstset\cup \secondset)}{\gOPT_{\hat{t}}} ~&<~ \frac{2\rho \firstconstant+2 \rconstant-\rho u+\frac{\rho}{2^{\hat{t}-2}}}{2\rho -2} 
\end{align}

In summary, the constraints that we get from the assumption that $v(\firstset)$, $v(\secondset)$, and $v(\thirdset)$ are not high enough are the following three:
\begin{align*}
   \frac{v(\firstset\cup \secondset)}{\gOPT_{\hat{t}}} &~>~ \rho \firstconstant -\rconstant.\\
   \frac{v(\firstset\cup \secondset)}{\gOPT_{\hat{t}}} &~>~ \rho (1-u)-\rconstant.\\
   \frac{v(\firstset\cup \secondset)}{\gOPT_{\hat{t}}} &~<~ \frac{2\rho \firstconstant+2 \rconstant-\rho u+\frac{\rho}{2^{\hat{t}-2}}}{2\rho -2}.
\end{align*}

The rest of the proof shows that these three constraints are incompatible, leading to a contradiction. We first show that it suffices to prove the constraints are incompatible when $u=1-\firstconstant$: we break into casework depending on whether $\firstconstant\ge 1-u$ or if $\firstconstant\le 1-u$.

     \textbullet~ For the first case, assume $\firstconstant \geq 1-u$. Then the second constraint becomes redundant and the third constraint is least restrictive when $u$ is minimized, we can therefore without loss of generality assume that $u=1-\firstconstant$.

    \textbullet~ Now assume that $\firstconstant \leq 1-u$. In this case the first lower bound becomes redundant and we can focus on the second lower bound and the upper bound. Take any values of $\firstconstant, \rconstant$, and $u$ and let $\mathcal{L}$ be the value of the lower bound and $\mathcal{U}$ be the value of the upper bound for this choice of $\firstconstant, \rconstant$, and $u$. If the constraints are compatible, i.e., $\mathcal{L}\leq \mathcal{U}$, then we note that increasing the value of $u$ by some $\delta>0$ (while keeping $\firstconstant$ and $\rconstant$ fixed), then the new lower bound would be equal to $\mathcal{L}-\rho \delta$ and the new upper bound would be $\mathcal{U}-\frac{\rho}{2\rho-2}\delta$. It is easy to verify that as long as $\rho> 1.5$, the distance between the upper bound the lower bound would increase. From prior work, we know that an approximation factor of $\sqrt{2}+1$  or better is not possible even for additive valuations \cite{chen2011approximability}, so we can safely assume that $\rho > 1.5$. Therefore, the constraints are least restrictive if we let $u$ take the largest possible value which, for this case, is once again equal to $1-\firstconstant$.
    
    With the case analysis above, we have shown that proving that the constraints are incompatible when $u=1-\firstconstant$ implies that they are incompatible in general.
    Substituting $u = 1-\firstconstant$ in the upper bound, it becomes $\frac{3\rho \firstconstant+2 \rconstant-\rho+\frac{\rho}{2^{\hat{t}-2}}}{2\rho -2}$. Combining it with the lower bound, we get
    \begin{align*}
    \rho \firstconstant -\rconstant &< \frac{3\rho \firstconstant+2 \rconstant-\rho+\frac{\rho}{2^{\hat{t}-2}}}{2\rho -2} ~~\Rightarrow \\
    2\rho^2 \firstconstant-2\rho \rconstant -2\rho \firstconstant +2\rconstant &< 3\rho \firstconstant+2 \rconstant -\rho+\frac{\rho}{2^{\hat{t}-2}} ~~\Rightarrow\\
    (2\rho^2-5\rho) \firstconstant-2\rho \rconstant &< -\rho+\frac{\rho}{2^{\hat{t}-2}}
    \end{align*}
    
    From Lemma~\ref{lem:lastelement} we get $  v(\firstset) \geq\frac{2^{\hat{t}-2}}{2^{\hat{t}-2}+1}\gOPT_{\hat{t}-1}=\frac{2^{\hat{t}-2}}{2^{\hat{t}-1}+2}\gOPT_{\hat{t}}$, therefore $\firstconstant \geq \frac{2^{\hat{t}-2}}{2^{\hat{t}-1}+2}$. We also have that $\marginal{R}{\firstset \cup \secondset} \leq \left(\frac{3}{2}- \frac{1}{2^{\hat{t}-1}}\right)\gOPT_{\hat{t}}$ from Lemma~\ref{lem:main}, i.e., $\rconstant \leq \frac{3}{2}- \frac{1}{2^{\hat{t}-1}}$. Substituting the lower bound for $\firstconstant$ and the upper bound for $\rconstant$ we get
    \begin{align*}
    (2\rho^2-5\rho) \left(\frac{2^{\hat{t}-2}}{2^{\hat{t}-1}+2}\right)-2\rho \left(\frac{3}{2}-\frac{1}{2^{\hat{t}-2}}\right) &< -\rho+\frac{\rho}{2^{\hat{t}-2}} ~~\Rightarrow\\
    (2\rho^2-5\rho) 2^{2\hat{t}-4}-2\rho \left(3(2^{\hat{t}-1}+2)2^{\hat{t}-3}-(2^{\hat{t}-1}+2)\right) &< -\rho(2^{\hat{t}-1}+2)2^{\hat{t}-2}+\rho (2^{\hat{t}-1}+2) ~~\Rightarrow\\
    (2\rho^2-5\rho) 2^{2\hat{t}-4}-2\rho \left(3(2^{2\hat{t}-4}+2^{\hat{t}-2})-(2^{\hat{t}-1}+2)\right) &< -\rho(2^{2\hat{t}-3}+2^{\hat{t}-1})+\rho (2^{\hat{t}-1}+2) ~~\Rightarrow\\
    (2\rho^2-9\rho) 2^{2\hat{t}-4}- 2^{\hat{t}-1}\rho+2\rho &< 0.
    \end{align*}
Further simplifying the inequality, 
    \begin{align*}
        (2\rho^2-9\rho) 2^{2\hat{t}-4} < 2^{\hat{t}-1}\rho-2\rho &~~\Rightarrow~~ (2\rho-9) 2^{2\hat{t}-4} < 2^{\hat{t}-1}-2\\
        2\rho - 9 < \frac{1}{2^{\hat{t}-3}} - \frac{1}{2^{2\hat{t}-5}} &~~\Rightarrow~~ \rho < \frac{1}{2^{\hat{t}-2}} - \frac{1}{2^{2\hat{t}-4}}+4.5,
    \end{align*}
    where the term $\frac{1}{2^{\hat{t}-2}} - \frac{1}{2^{2\hat{t}-4}}$ is maximized at $\hat{t}=3$, at which we have:
    \[\rho < 0.25+4.5 = 4.75.\]
    
    Therefore, for any $\rho \geq 4.75$ it is impossible to satisfy all the three constraints listed above, leading to a contradiction. 
\end{proof}

\subsection{The Analysis of the Approximation for $\hat{t}\leq 2$}\label{subsec:teq2}
\begin{lemma}\label{lem:teq2}
If $\hat{t} \leq 2$, Then \textsc{Iterative-pruning} auction would guarantee a 4 approximation of the optimal value.
\end{lemma}
\begin{proof}
From Lemma \ref{lem:main} we bound the optimal solution as follows \[\OPT \leq v(\firstset \cup \secondset) + \marginal{R}{\firstset \cup \secondset} \leq 2\max\{v(\firstset), v(\secondset)\}+\marginal{R}{\firstset \cup \secondset}\] 

 When $\hat{t}=2$, $\firstset$ is budget feasible without removing any of its bidders since $\firstset$ is initialized to $S_1$ which is simply the maximum value item offered the entire budget. We can also assume $\secondset$ is budget feasible or else we get a 3.75 approximation by Appendix \ref{subsec:W2budgetfeasible}. Thus we can simply pick the greater value set out of $\firstset$ and $\secondset$, so to get a 4 approximation it suffices to show 

 \[\frac{2\max\{v(\firstset), v(\secondset)\}+\marginal{R}{\combinedset}}{\max\{v(\firstset), v(\secondset)\}} \le 4  \]
 
 In this case, following the analysis of  $\marginal{R}{\firstset \cup \secondset}$ in Lemma \ref{lem:main} we get a stronger bound of $\marginal{R}{\combinedset} \le \gOPT_{\hat{t}}$ when $\hat{t}=2$. This is because the auction only consists of two phases so we can omit the $(\frac{1}{2}-\frac{1}{2^{\hat{t}-2}})\gOPT_{\hat{t}}$ loss from $\R_a$ from rounds before $\hat{t}-1$. We also know $v(\firstset) = \frac{1}{2}\gOPT_{\hat{t}}$ since no bidders had to be removed from $\firstset$ to make it budget feasible. Therefore for $\hat{t}=2$, 
 
 \[\frac{2\max\{v(\firstset), v(\secondset)\}+\marginal{R}{\combinedset}}{\max\{v(\firstset), v(\secondset)\}} \le \frac{\gOPT_{\hat{t}}+\gOPT_{\hat{t}}}{\frac{1}{2}\gOPT_{\hat{t}}} = 4 \]
 
\end{proof}

\subsection{The Analysis of the Approximation for Non-Budget-Feasible $\bsecondset$}\label{subsec:W2budgetfeasible}
\begin{lemma}\label{lem:budgetW2}
If $\bsecondset$ is not budget feasible, then \textsc{Iterative-Pruning} auction would guarantee a 3.75 approximation of the optimal value by outputting $W_3$.
\end{lemma}

\begin{proof}
First note that by definition we have $v(W_3) \geq v(\secondset)$.
From Lemma~\ref{lem:main} we can upper bound the optimal solution as follows:
\[\OPT \leq v(\firstset \cup \bsecondset) + \marginal{R}{\firstset \cup \bsecondset} \leq v(\firstset) + v(\bsecondset) + \marginal{R}{\firstset \cup \bsecondset}\]

To upper bound the value of $\bsecondset$, notice that the maximum possible value of $v(\bsecondset)$ is obtained if, during the construction of $S_{\hat{t}}$, adding the last bidder makes $\bsecondset$ exceed the budget. Further, the last bidder added to $S_{\hat{t}-1}$, $j^*$, accepted the new price and was added to $\bsecondset$. Removing the two bidder from $\bsecondset$ we get $\secondset$.  The largest budget-feasible prefix of $\bsecondset$, is bounded by $\gOPT_{\hat{t}}$, i.e., $v(\secondset) < \gOPT_{\hat{t}}$. Note that for every bidder $i$ we have $v(\{i\}) \leq \gOPT_1 = \frac{\gOPT_t}{2^{t-1}}$ for every $t\geq 1$. Then the value of $\bsecondset$ should be no more than $\secondset$ combined with two extra bidders, by submodularity the value of $\bsecondset$ is at most
\[v(\bsecondset) \leq v(\secondset)+2\gOPT_1 < \left(1+\frac{1}{2^{\hat{t}-2}}\right)\gOPT_{\hat{t}}. \]

By lemma~\ref{lem:lastelement} we have:

\[v(W_3) \geq v(\secondset) \geq \frac{2^{\hat{t}-1}}{2^{\hat{t}-1}+1}S_{\hat{t}} \geq  \frac{2^{\hat{t}-1}}{2^{\hat{t}-1}+1}\gOPT_{\hat{t}}.\]

And by lemma~\ref{lem:main}, we also get the the marginal contribution of $R$ is
\[\marginal{R}{\firstset \cup \bsecondset} \leq \left(\frac{3}{2}-\frac{1}{2^{\hat{t}-2}}\right)\gOPT_{\hat{t}}.\]

Therefore the approximation would be:

\[\frac{\OPT}{v(W_3)}\leq  \frac{v(\firstset) + v(\bsecondset) + \marginal{R}{\firstset \cup \bsecondset}}{v(\secondset)} < \frac{\frac{1}{2}+1+\frac{1}{2^{\hat{t}-2}}+\frac{3}{2}-\frac{1}{2^{\hat{t}-2}}}{\frac{2^{\hat{t}-1}}{2^{\hat{t}-1}+1}}~~\Rightarrow\]

\[\frac{\OPT}{v(W_3)}< \frac{3}{\frac{2^{\hat{t}-1}}{2^{\hat{t}-1}+1}} \leq 3.75.\]
for any $\hat{t}\geq 3$.
\end{proof}

\subsection{Budget feasibility of \mechone}\label{subsec:budegtfeasible}
\begin{lemma} 
\label{lem:budgetfeasible}
\mechone \ is a budget-feasible mechanism.
\end{lemma}

\begin{proof}
Since we return either $\firstset,\secondset,$ or $\thirdset$, it suffices to show each set is budget feasible. 

$\firstset$ is initialized to $S_{\hat{t}-1}$ at the end of the last phase. If the last bidder $\lastbidder$ added to $\firstset$ during phase $\hat{t}-1$ was not removed from $\firstset$ by the mechanism, then by the condition of the if statement, 
$\sum_{i \in \firstset} p_i \leq B$ and $\firstset$ is budget feasible.

Otherwise,  $\firstset = S_{\hat{t}-1}\setminus \{\lastbidder\}$ and, by the condition of the inner while loop, we have $v(S_{\hat{t}-1} \setminus \{\lastbidder\}) < \gOPT_{\hat{t}-1}$. Let $S_{\hat{t}-1}^i$ be the set $\Scurr$ at the beginning of the iteration of phase $\hat{t}-1$ where bidder $i$ is considered. The prices $p^{\star}_i$ paid to each $i \in \firstset$ are the last price they were offered, so  we have $p^{\star}_i \leq p_i^t = \marginal{i}{S_{\hat{t}-1}^i} \cdot \frac{B}{\gOPT_{\hat{t}-1}}$ where $p_i^{\hat{t}-1}$ is the price $p_i$ at phase ${\hat{t}-1}$. We get that
$$\sum_{i \in \firstset  }p^{\star}_i \leq \sum_{i \in \firstset \setminus \{\lastbidder\} } \marginal{i}{\firstset} \cdot \frac{B}{\gOPT_{\hat{t}-1}} = v(\firstset \setminus \{\lastbidder\} ) \cdot \frac{B}{\gOPT_{\hat{t}-1}} \leq B$$
and $\firstset$ is budget feasible.

$\secondset$ is budget feasible by definition as it is the largest budget feasible prefix of $\bsecondset$. Similarly $\thirdset$ is also budget feasible by definition as it is $\secondset$ joined with the largest prefix of $\firstset$ that will not exceed the available budget not used by $\secondset$. Thus each of $\firstset, \secondset$, and $\thirdset$ are budget feasible sets, making \textsc{Iterative-Pruning} a budget feasible mechanism. 
\end{proof}

\subsection{Running time of \mechone}\label{subsec:polytimemonotone}
\begin{lemma} 
\label{lem:polytimemonotone}
\mechone has $O(n^2 \log{n})$ running time.
\end{lemma}
\begin{proof} 
We first compute the time of completing one round of the auction (i.e., one iteration of the outer while loop) and then upper bound the total number of rounds.  
 Observe that the inner while loop completes at most $n$ times per iteration of the outer while loop since a bidder is either removed from $A$ or added to $S_{t}$ on each iteration.  The body of the inner while loop completes in $O(n)$ time as at most $n$ value queries are made to find the bidder of largest marginal contribution.  Thus, in total, lines \ref{algline:startouter} through \ref{algline:endouter} take $O(n^2)$ time per iteration of the outer while loop.  

We now move to bound the number of iterations of the outer while loop.   Observe that if at some iteration of the outer while loop the estimate $\gOPT$ exceeds the actual value of the optimal solution then the inner while loop will only terminate when $A \setminus \left(S_{t-1} \cup S_{t}\right) = \emptyset$, which, by consequence, will also terminate the outer while loop.  However, we have that $\gOPT$ begins as the single highest value of any individual bidder and, by submodularity, $\OPT$ is at most $n$ times this initial value.  Thus, since $\gOPT$ doubles in each round, there are at most $O(\log{n})$ iterations of the outer while loop.  Thus, in total, the first phase of the auction terminates in $O(n^2 \log{n})$ time.

Since the pruning phase (lines \ref{algline:prunestart} through \ref{algline:pruneend} and \textsc{Maximize-Value}) completes in $O(n)$ time (even if the sum of clock prices needs to be computed again),   the entirety of the auction completes in $O(n^2 \log{n})$ time.
\end{proof}

\subsection{Proof of Lemma~\ref{lem:lowerboundInstance}}\label{subsec:lowerbound}
\begin{proof}
Consider an instance with the following 4 sets of bidders $A_1, A_2, A_3, A_4$.  $A_1$ consists of a single bidder $i_1$ with $v(\{i_1\})=1$ and $c_{i_1}=B$. $A_2$ consists of 3 bidders $i_2,i_3,i_4$  where $v(\{i_2\})=v(\{i_3\})=v(\{i_4\})=\frac{2}{3}+\epsilon$. The costs of the bidders in $A_2$ are $c_{i_2}=c_{i_3}=0, c_{i_4}=\frac{(2/3+\epsilon)B}{2}$.
$A_3$ consists of $\frac{4}{3\epsilon}$ identical bidders with value $\epsilon$ and cost 0. Finally $A_4$ consists of $\frac{8}{\epsilon}$ bidders with value $\frac{\epsilon}{2}$ and cost $\frac{(\epsilon/2+\delta)B}{4}$ where $\delta \ll \epsilon$. All of the bidders have additive value with each other except for $i_2$ and the bidders in $A_3$ which are ``capped additive'' (i.e., budget additive) with cap $4/3 = v(A_3)$.  In other words, for any $S \subseteq A_3 \cup \{i_2\}$ we have $v(S) = \min\left\{\sum_{i \in S}{v(\{i\})}, 4/3\right\}$ and for any output set $S \subseteq A_1 \cup A_2 \cup A_3 \cup A_4$ we have that $v(S) = \sum_{i \in S; i \notin A_3 \cup \{i_2\}}{v(\{i\})} + \min\left\{\sum_{i \in S; i \in A_3 \cup \{i_2\}}{v(\{i\})}, 4/3\right\}$.

Now run \textsc{Iterative-Pruning} with a budget $B$ and these bidders. The mechanism would start by initializing $S_1$ to $A_1$ with $\gOPT_1 = 1$ since $i_1$ has the highest individual value out of all the bidders. Then the mechanism would set $\gOPT_2=2$ and approach all of the bidders in $A_2$ offering a price of $\frac{(2/3+\epsilon)B}{2}$ to each of them leading to all of them being accepted to $S_2$. 

Since $v(A_2)\ge \gOPT_2$, the mechanism would set $\gOPT_3 = 4$ and move onto constructing $S_3$. First $i_1$ would be offered a price of $\frac{B}{4}$ causing it to reject. Then the auction would approach all of the bidders in $A_3$ with a price of $\frac{\epsilon B}{4}$ causing them all to accept and be added to $S_3$. Finally the mechanism would approach each bidder in $A_4$ with a price of $\frac{\epsilon B}{8}$ causing them all to reject since $\frac{(\epsilon/2+\delta)B}{4} >\frac{\epsilon B}{8}$. Since every item not in $S_2$ or $S_3$ has been rejected at this point, the initial while loop concludes giving us $W_1 = S_2 = A_2$ and $W_2 = S_3 = A_3$. By offering each of $i_2,i_3,i_4$ a price of $\frac{(2/3+\epsilon)B}{2}$ we exceed the total budget for $W_1$ and accordingly update the price of $i_4$ to $\frac{(2/3+\epsilon)B}{4}$. Then $i_4$ rejects this price so we are left with $W_1 = \{i_2\}\cup \{i_3\}$ and $W_2 = A_3$.

Moving to the \textsc{Maximize-Value} subroutine, we have $T = \{i_2\}$ so $W_3 = \{i_2\} \cup A_3$. Thus we are left with $v(W_1) = 4/3 + 2\epsilon$, $v(W_2) = 4/3$ and $v(W_3) = 4/3$ leading the mechanism to return $W_1$.  However, the optimal budget feasible solution consists of $\{i_2\}\cup \{i_3\} \cup A_3 \cup A_4 \setminus \{i^-\}$ where $i^-$ is the last item in $A_4$ since all of the bidders in $\{i_2\}\cup \{i_3\}\cup A_3$ have 0 cost and the cost of $A_4/\{i^-\}$ is $B(1+\frac{2\delta}{\epsilon} - \frac{\epsilon}{8}+\frac{\delta}{4}) < B$. $i_1,i_4$ have much worse marginal densities per cost than the bidders in $A_4$ so they are left out of $\OPT$ in favor of $A_4$. Thus we have $\OPT = v(\{i_2\}\cup \{i_3\} \cup A_3 \cup A_4) = \frac{2}{3}+\epsilon + \frac{4}{3} + 4 - \epsilon/2$ giving an approximation factor of $\frac{6+\epsilon/2}{4/3+2\epsilon}$ which is no better than a 4.5 approximation for arbitrary $\epsilon$.
\end{proof}

\section{Limitations of posted-price mechanisms}\label{sec:postedprice}

The existence of deterministic budget feasible clock auctions that achieve a constant approximation raises the question of whether there exist even simpler families of budget feasible mechanisms with which  one can obtain constant approximations mechanisms. In this section, we  study  deterministic posted-price mechanisms, which are arguably the simplest family of mechanisms. We show that even for the special cases of additive valuation functions (Section~\ref{sec:additive}) and symmetric valuation functions (Section~\ref{sec:symmetric}), there are no deterministic posted-price mechanisms that  achieve a constant approximation.
Recall that posted-price mechanisms approach sellers in some order and make ``take-it-or-leave-it'' offers.  In other words, a posted-price mechanism offers each seller $i$ a single price $p_i$ (the price offers can differ for each seller), which $i$ accepts if $p_i \geq \costi$ and rejects otherwise.   For a posted-price mechanism to be budget feasible, the sum of the prices of the accepted offers must not exceed the budget.

\subsection{Additive Valuation Functions}
\label{sec:additive}
We first examine the special case of additive valuation functions.  A valuation function $v$ is additive if for all $S \subseteq \bidders$ we have that $v(S) = \sum_{i \in S}{v_i}$. We show  that deterministic posted-price mechanisms cannot achieve an approximation factor better than $\Omega\left(\sqrt{n}\right)$.

\begin{theorem}
No deterministic posted-price mechanism can achieve an approximation better than $\sqrt{n}/2$ for instances with  additive valuation functions.
\end{theorem}

\begin{proof}
We consider a family of instances with $n$ bidders where there is a single bidder with value $\sqrt{n}$  and each remaining bidder has value $1$.  We denote the high value bidder $b_h$ and in each instance $\costi[b_h] = B$.  The family of instances differ only on the costs of the small value bidders.  We perform case analysis on the offers that any mechanism makes to the bidders. 

\paragraph{Case 1}
Suppose the mechanism $\mathcal{M}$ offers some positive price to a small value bidder before it makes an offer to bidder $b_h$.  Let $i$ denote the first small value bidder the mechanism makes a positive offer $p_i$ to. Consider the instance where $c_i = p_i$ and all the other sellers have cost $B$. Then, regardless of the other offers, to maintain budget feasibility the mechanism can only obtain value $1$, where the optimal solution is to output bidder $b_h$ and obtain value $\sqrt{n}$. We then have:
\[\alpha \geq \frac{\OPT}{v(\mathcal{M})} = \frac{\sqrt{n}}{1}> \sqrt{n}/2. \]

\paragraph{Case 2} Suppose the mechanism $\mathcal{M}$ offers each bidder before $b_h$ price $0$ and offers the entire budget $B$ to bidder $b_h$. Consider the instance where all the small bidders have cost $\frac{1}{n-1}$.  The mechanism will then obtain total value equal to $\sqrt{n}$.  On the other hand, the optimal solution would be to output the $n-1$ smaller bidders and the value would be $n-1$. We then have:
\[\alpha \geq 
\frac{\OPT}{v(\mathcal{M})} = \frac{n-1}{\sqrt{n}} > \sqrt{n}/2.\]

\paragraph{Case 3} Suppose the mechanism $\mathcal{M}$ offers each bidder before $b_h$ price $0$ and offers price $p_h < B$ to bidder $b_h$.  Consider the instance where the cost of each seller is $B$.  Then the mechanism can only obtain value $1$, whereas the optimal solution is to output bidder $b_h$, obtaining a value of $\sqrt{n}$.  We then have:
\[\alpha \geq \frac{\OPT}{v(\mathcal{M})} = \frac{\sqrt{n}}{1}> \sqrt{n}/2. \qedhere \] 
\end{proof}

\subsection{Symmetric Submodular Valuation Functions}
\label{sec:symmetric}

We now consider another special subclass of submodular valuation functions. 
A function $v : 2^{\bidders} \rightarrow \mathbb{R}^{\geq 0}$ is \emph{symmetric} submodular if there exist $r_1 \geq r_2 \geq \dots \geq r_n \geq 0$, such that $v(S) = \sum_{i = 1}^{|S|} r_i$ for all $S \subseteq 2^\bidders$.  This class of functions was studied in the work of \citet{vickrey1961counterspeculation} on multi-unit auctions and was studied in the context of budget feasible procurement by \citet{singer2010budget} and \citet{BKS12}.  We show that within this restricted family of instances, where the goal of the auctioneer is to maximize the number of sellers that accept the prices offered to them, no deterministic posted-price mechanism can achieve a constant approximation. \footnote{We also note that our lower bound also applies to symmetric \emph{additive} valuations, i.e., where $r_i = 1$ for all $i \in [n]$}
\begin{theorem}\label{lem:PostedLowerBound}
No deterministic posted-price mechanism can achieve an approximation factor better than $\frac{\log n}{4}$ for instances with symmetric submodular valuation functions.
\end{theorem}

\begin{proof}
\begin{figure}
    \centering
\tikzset{ text width =30cm, every node/.style={scale=.8}} 
\begin{tikzpicture}[x=0.75pt,y=0.75pt,yscale=-1,xscale=1]

\node[] (Instance1) at (0,0) {Instance 1 \ \ \ \ \ $\underbrace{\frac{4B}{\log n}, \frac{4B}{\log n},\cdots ,\frac{4B}{\log n}}_{\textstyle\#=\frac{\log n}{4}} \ \ \ \ \   B,B,\dots,B$};

\node[] (Instance2) at (0,70) {Instance 2 \ \ \ \ \ $\underbrace{\frac{4B}{\log n}, \frac{4B}{\log n},\cdots ,\frac{4B}{\log n}}_{\textstyle\#=\frac{\log n}{4}} \ \ \ \  \underbrace{\frac{2B}{\log n},\frac{2B}{\log n},\dots,\frac{2B}{\log n}}_{\textstyle\#=\frac{\log n}{2}} \ \ \ \ \  B,B,\dots,B$};

\node[] (Instancek) at (0,140) {Instance k \ \ \ \ \ $\underbrace{\frac{4B}{\log n}, \frac{4B}{\log n},\cdots ,\frac{4B}{\log n}}_{\textstyle\#=\frac{\log n}{4}} \ \ \ \underbrace{\frac{2B}{\log n},\frac{2B}{\log n},\dots,\frac{2B}{\log n}}_{\textstyle\#=\frac{\log n}{2}} \  \dots \  \underbrace{\frac{B}{2^k\log n},\frac{B}{2^k\log n},\dots,\frac{B}{2^k\log n}}_{\textstyle\#=2^{k-3}\log n}\ \ \ \ \   B,B,\dots,B$};

\node[] (Dots) at (25,90) {\Huge$\vdots$};
\end{tikzpicture}
    \caption{Seller costs for the instances used in the construction of Lemma~\ref{lem:PostedLowerBound}}
    \label{fig:cpplot}
\end{figure}

Consider a specific symmetric function $v(S) = |S|$.  We define a family of instances, where in instance $k$, we partition the bidders into $k+1$ groups. Each of the $j \leq k$ groups contains  $2^{j-3} \log n $ bidders each with cost $\frac{B}{2^{j-3}\log n}$.  Notice that each group is budget feasible since $2^{j-3}\log n \cdot \frac{B}{2^{j-3}\log n} = B$. Then the $k+1$-th group contains all the remaining bidders each with cost $B$. For example, instance 1 has $\frac{\log n}{4}$ bidders with cost $\frac{4B}{\log n}$, and all the remaining bidders have a cost of $B$. We can see that the optimal solution in instance $k$ is to output all sellers in group $k$, let $\OPT_k$ denote the optimal value of instance $k$, we have:
\[\OPT_k = 2^{k-3} \log n.\]

In order to achieve the $\frac{\log n}{4}$ approximation factor in any instance $k$, a mechanism $\mathcal{M}$ needs to output at least $\frac{\OPT}{\log n /4} \geq 2^{k-1}$ sellers in instance $k$. Therefore to simultaneously achieve the $\frac{\log n}{4}$ approximation factor in instance 1 through $k$, $\mathcal{M}$ needs to output $2^{j-1}$ bidders in instance $j$ for all $j \in [1,k]$ and the minimum amount the mechanism needs to pay is then
\[ \frac{4B}{\log n} + \sum_{j = 2}^{k} 2^{j-2} \cdot \frac{4B}{2^{j-1}\log{n}}= \frac{2B}{\log n} \cdot (k+1)\]
by purchasing exactly $2^{j-1}-2^{j-2} = 2^{j-2}$ bidders from group $j$.\\

Solving $\frac{2B}{\log n} \cdot (k+1) = B$ we get that $k = \frac{\log n}{2} -1 = \log \sqrt{n}-1$. In other words, to satisfy $\log{\sqrt{n}-1}$ instances, we need to use all of our budget. Now let $n_o$ be the total number of bidders with cost less than $B$ in the $\log{\sqrt{n}} - 1$-th instance.  We then have:

\[ n_o =  \sum_{j = 1}^{\log \sqrt{n}-1} 2^{j-3} \log{n}, \]
and by geometric sum we have:
\[n_o = \frac{\log n}{4} \cdot \frac{2^{\log \sqrt{n}-1}-1}{2-1} = \frac{ \sqrt{n}/2 \log n - \log n}{4} < \frac{\sqrt{n}\log{n}}{8}\]
Now consider a instance with $\log{\sqrt{n}}-1$ groups as we defined, (each group $j \in [1, \log{\sqrt{n}}]$ has $2^{j-3}\log{n}$ sellers and each seller costs $\frac{B}{2^{j-3}\log{n}}$),  and remaining sellers all have a cost $\frac{B}{n-n_o}$, we get:
\[\OPT =n - n_o > n - \frac{\sqrt{n} \log n}{8}\]

However, the mechanism would have used up the budget in the process of guaranteeing the approximation factor in the $\log{\sqrt{n}}$ instances we defined, therefore it has no remaining budget to purchase any seller.  But then, the total value the mechanism must obtain is
\[v(\mathcal{M})= 2^{\log{\sqrt{n}}-2} = \frac{\sqrt{n}}{4}.\]
Therefore, the approximation factor is at least:
\[\alpha = \frac{\OPT}{v(\mathcal{M})} \geq \frac{ n - \sqrt{n} \log n/8}{\sqrt{n}/4} = 4\sqrt{n} - \frac{\log n}{2} > \frac{\log n}{4}\]
completing the proof (since the valuation function $v$ is a symmetric submodular function).
\end{proof}
We now present a mechanism that achieves a $O(\log{n})$ approximation to the optimal welfare with symmetric submodular valuations.  Note that we assume that $c_i \leq B$ for all agents $i$.

    \begin{algorithm}[H]\label{alg:PostedSymmetricSub}
    \SetKwInOut{Input}{Input}
    \Input{A public budget $B$, and an arbitrarily ordered set of bidders $\{i\}_{[n]}$ with private costs $c_i$\; a public additive valuation function $\V$.}
    Initialize $a\leftarrow0$\;
    Initialize $W \leftarrow \emptyset$\;
    set aside an arbitrary agent $j$\;
    \For{ $i \in [n]\setminus \{j\}$}{
    \If{ $a = 0$}{
    Offer $p_i \leftarrow \frac{B}{2\ln n}$ to agent $i$\;
     }
        \Else{
        Offer $p_i \leftarrow \frac{B}{a\cdot 2\ln n}$ to agent $i$\;}
    \If{\text{agent} i accepts}{
        $a \leftarrow a+1$\;
        $W \leftarrow W \cup \{i\}$}
        }
    \If{a = 0}{
        Offer $p_j \leftarrow B$ to agent $j$\;}
    \Return W
    \caption{A posted-price mechanism for symmetric submodular valuations.}
    \end{algorithm}

\begin{theorem}
Mechanism~\ref{alg:PostedSymmetricSub} obtains a $O(\log{n})$ approximation to the optimal value for instances with symmetric submodular valuations. 
\end{theorem}
\begin{proof}
First, if all the agents accept our offer, the total payment would be 
\[(1 + 1+ 1/2 + 1/3 + \dots + 1/(n-1)) \frac{B}{2\ln{n}} \leq \frac{\ln{(n-1)}+2}{2\ln{n}}B < B\]
Therefore, at anytime of the execution of the mechanism, we would not exhaust the budget. Let $r_i$ be the marginal gain of adding the the $i^{\text{th}}$ agent to the winning set. Now consider the following cases:

\paragraph{Case 1}If no agent accepts the offer, the mechanism $\mathcal{M}$ would return the agent $j$ giving $v(\mathcal{M}) = r_1$, since each agent is offered and then rejected at a price of $\frac{B}{2\ln{n}}$, we have for each agent $i$, $c_i > \frac{B}{2\ln{n}}$. We can fit at most $2 
\ln n +1 $ more agents. $\OPT \leq \sum_{i = 1}^{2\ln n+1} r_i$. By the definition of symmetric submodular, $r_1 \geq r_2 \geq \dots \geq r_n$, we have $\OPT \leq \sum_{i = 1}^{2\ln n+1} r_i \leq (2 \ln n +1) \cdot r_1$
\[\alpha = \frac{\OPT}{v(\mathcal{M})} =\frac{(2\ln{n}+1)\cdot r_1}{r_1} < 2\ln{n}+1.\] 

\paragraph{Case 2}Now let $k$ be the number of agent returned by $\mathcal{M}$, we first have $v(\mathcal{M}) = \sum_{i = 1}^{k} r_i \geq k \cdot r_k$. Now for any agent $j$ rejected after the $k^{th}$ accepted agent, the price offered is $\frac{B}{k \cdot 2\ln{n}}$, therefore we have that $c_j > p_j = \frac{B}{k\cdot 2\ln{n}}$. For any agent $i$ rejected before agent $k$, we have $c_i > p_i > p_j = \frac{B}{k\cdot 2\ln{n}}$, therefore the optimal solution can at fit less than $k\cdot 2\ln{n}+k$ agents. $\OPT < \sum_{i = 1}^{k\cdot2\ln n+k} r_i \leq v(\mathcal{M}) + (k\cdot2 \ln n)r_k$. We have:
\[\alpha = \frac{\OPT}{v(\mathcal{M})} = \frac{v(\mathcal{M}) + (k\cdot2 \ln n)r_k}{v(\mathcal{M})} = 1 + \frac{(k\cdot2 \ln n)r_k}{v(\mathcal{M})}  < 1+\frac{(k\cdot2 \ln n)r_k}{k \cdot r_k} = 2\ln{n}+1\qedhere\]
\end{proof}

\end{appendices}

\newpage

\bibliographystyle{plainnat}
\bibliography{bibliography}

\begin{thebibliography}{}

\end{thebibliography}


\begin{thebibliography}{34}
\providecommand{\natexlab}[1]{#1}
\providecommand{\url}[1]{\texttt{#1}}
\expandafter\ifx\csname urlstyle\endcsname\relax
  \providecommand{\doi}[1]{doi: #1}\else
  \providecommand{\doi}{doi: \begingroup \urlstyle{rm}\Url}\fi

\bibitem[Amanatidis et~al.(2019)Amanatidis, Kleer, and
  Sch{\"a}fer]{amanatidis2019budget}
Georgios Amanatidis, Pieter Kleer, and Guido Sch{\"a}fer.
\newblock Budget-feasible mechanism design for non-monotone submodular
  objectives: Offline and online.
\newblock In \emph{Proceedings of the 2019 ACM Conference on Economics and
  Computation}, pages 901--919, 2019.

\bibitem[Anari et~al.(2014)Anari, Goel, and Nikzad]{anari2014mechanism}
Nima Anari, Gagan Goel, and Afshin Nikzad.
\newblock Mechanism design for crowdsourcing: An optimal 1-1/e competitive
  budget-feasible mechanism for large markets.
\newblock In \emph{2014 IEEE 55th Annual Symposium on Foundations of Computer
  Science}, pages 266--275. IEEE, 2014.

\bibitem[Ausubel et~al.()Ausubel, Milgrom, et~al.]{ausubel2006lovely}
Lawrence~M Ausubel, Paul Milgrom, et~al.
\newblock The lovely but lonely {V}ickrey auction.

\bibitem[Babaioff et~al.(2020)Babaioff, Immorlica, Lucier, and
  Weinberg]{babaioff2014simple}
Moshe Babaioff, Nicole Immorlica, Brendan Lucier, and S.~Matthew Weinberg.
\newblock A simple and approximately optimal mechanism for an additive buyer.
\newblock \emph{J. {ACM}}, 67\penalty0 (4):\penalty0 24:1--24:40, 2020.

\bibitem[Badanidiyuru et~al.(2012{\natexlab{a}})Badanidiyuru, Dobzinski, and
  Oren]{badanidiyuru2012optimization}
Ashwinkumar Badanidiyuru, Shahar Dobzinski, and Sigal Oren.
\newblock Optimization with demand oracles.
\newblock In \emph{Proceedings of the 13th ACM conference on electronic
  commerce}, pages 110--127, 2012{\natexlab{a}}.

\bibitem[Badanidiyuru et~al.(2012{\natexlab{b}})Badanidiyuru, Kleinberg, and
  Singer]{BKS12}
Ashwinkumar Badanidiyuru, Robert Kleinberg, and Yaron Singer.
\newblock Learning on a budget: posted price mechanisms for online procurement.
\newblock In \emph{Proceedings of the 13th ACM conference on electronic
  commerce}, pages 128--145, 2012{\natexlab{b}}.

\bibitem[Balkanski and Hartline(2016)]{BH16}
Eric Balkanski and Jason~D Hartline.
\newblock Bayesian budget feasibility with posted pricing.
\newblock In \emph{Proceedings of the 25th International Conference on World
  Wide Web}, pages 189--203, 2016.

\bibitem[Bei et~al.(2017)Bei, Chen, Gravin, and Lu]{bei2017worst}
Xiaohui Bei, Ning Chen, Nick Gravin, and Pinyan Lu.
\newblock Worst-case mechanism design via {B}ayesian analysis.
\newblock \emph{SIAM Journal on Computing}, 46\penalty0 (4):\penalty0
  1428--1448, 2017.

\bibitem[Buchbinder et~al.(2015)Buchbinder, Feldman, Seffi, and
  Schwartz]{buchbinder2015tight}
Niv Buchbinder, Moran Feldman, Joseph Seffi, and Roy Schwartz.
\newblock A tight linear time (1/2)-approximation for unconstrained submodular
  maximization.
\newblock \emph{SIAM Journal on Computing}, 44\penalty0 (5):\penalty0
  1384--1402, 2015.

\bibitem[Chen et~al.(2011)Chen, Gravin, and Lu]{chen2011approximability}
Ning Chen, Nick Gravin, and Pinyan Lu.
\newblock On the approximability of budget feasible mechanisms.
\newblock In \emph{Proceedings of the twenty-second annual ACM-SIAM symposium
  on Discrete Algorithms}, pages 685--699. SIAM, 2011.

\bibitem[Chen et~al.(2018)Chen, Matikas, Paparas, and
  Yannakakis]{chen2018complexity}
Xi~Chen, George Matikas, Dimitris Paparas, and Mihalis Yannakakis.
\newblock On the complexity of simple and optimal deterministic mechanisms for
  an additive buyer.
\newblock In \emph{Proceedings of the Twenty-Ninth Annual ACM-SIAM Symposium on
  Discrete Algorithms}, pages 2036--2049. SIAM, 2018.

\bibitem[Dobzinski et~al.(2011)Dobzinski, Papadimitriou, and
  Singer]{dobzinski2011mechanisms}
Shahar Dobzinski, Christos~H Papadimitriou, and Yaron Singer.
\newblock Mechanisms for complement-free procurement.
\newblock In \emph{Proceedings of the 12th ACM conference on Electronic
  commerce}, pages 273--282, 2011.

\bibitem[D{\"u}tting et~al.(2017{\natexlab{a}})D{\"u}tting, Gkatzelis, and
  Roughgarden]{dutting2017performance}
Paul D{\"u}tting, Vasilis Gkatzelis, and Tim Roughgarden.
\newblock The performance of deferred-acceptance auctions.
\newblock \emph{Mathematics of Operations Research}, 42\penalty0 (4):\penalty0
  897--914, 2017{\natexlab{a}}.

\bibitem[D{\"u}tting et~al.(2017{\natexlab{b}})D{\"u}tting, Talgam-Cohen, and
  Roughgarden]{dutting2017modularity}
Paul D{\"u}tting, Inbal Talgam-Cohen, and Tim Roughgarden.
\newblock Modularity and greed in double auctions.
\newblock \emph{Games and Economic Behavior}, 105:\penalty0 59--83,
  2017{\natexlab{b}}.

\bibitem[Ensthaler and Giebe(2014)]{ensthaler2014dynamic}
Ludwig Ensthaler and Thomas Giebe.
\newblock A dynamic auction for multi-object procurement under a hard budget
  constraint.
\newblock \emph{Research Policy}, 43\penalty0 (1):\penalty0 179--189, 2014.

\bibitem[Feige(1998)]{feige1998threshold}
Uriel Feige.
\newblock A threshold of ln n for approximating set cover.
\newblock \emph{Journal of the ACM (JACM)}, 45\penalty0 (4):\penalty0 634--652,
  1998.

\bibitem[Gkatzelis et~al.(2017)Gkatzelis, Markakis, and
  Roughgarden]{gkatzelis2017deferred}
Vasilis Gkatzelis, Evangelos Markakis, and Tim Roughgarden.
\newblock Deferred-acceptance auctions for multiple levels of service.
\newblock In \emph{Proceedings of the 2017 ACM Conference on Economics and
  Computation}, pages 21--38, 2017.

\bibitem[Gravin et~al.(2020)Gravin, Jin, Lu, and Zhang]{gravin2019optimal}
Nick Gravin, Yaonan Jin, Pinyan Lu, and Chenhao Zhang.
\newblock Optimal budget-feasible mechanisms for additive valuations.
\newblock \emph{{ACM} Trans. Economics and Comput.}, 8\penalty0 (4):\penalty0
  21:1--21:15, 2020.
\newblock \doi{10.1145/3417746}.
\newblock URL \url{https://doi.org/10.1145/3417746}.

\bibitem[Horel et~al.(2014)Horel, Ioannidis, and
  Muthukrishnan]{horel2014budget}
Thibaut Horel, Stratis Ioannidis, and S~Muthukrishnan.
\newblock Budget feasible mechanisms for experimental design.
\newblock In \emph{Latin American Symposium on Theoretical Informatics}, pages
  719--730. Springer, 2014.

\bibitem[Jalaly and Tardos(2021)]{JT21}
Pooya Jalaly and {\'{E}}va Tardos.
\newblock Simple and efficient budget feasible mechanisms for monotone
  submodular valuations.
\newblock \emph{{ACM} Trans. Economics and Comput.}, 9\penalty0 (1):\penalty0
  4:1--4:20, 2021.

\bibitem[Jarman and Meisner(2017)]{JM17}
Felix Jarman and Vincent Meisner.
\newblock Ex-post optimal knapsack procurement.
\newblock \emph{Journal of Economic Theory}, 171:\penalty0 35--63, 2017.
\newblock ISSN 0022-0531.
\newblock \doi{https://doi.org/10.1016/j.jet.2017.06.001}.
\newblock URL
  \url{https://www.sciencedirect.com/science/article/pii/S0022053117300637}.

\bibitem[Kagel et~al.(1987)Kagel, Harstad, and Levin]{KHL87}
John~H Kagel, Ronald~M Harstad, and Dan Levin.
\newblock Information impact and allocation rules in auctions with affiliated
  private values: A laboratory study.
\newblock \emph{Econometrica: Journal of the Econometric Society}, pages
  1275--1304, 1987.

\bibitem[Kim(2015)]{kim2015welfare}
Anthony Kim.
\newblock Welfare maximization with deferred acceptance auctions in
  reallocation problems.
\newblock In \emph{Algorithms-ESA 2015}, pages 804--815. Springer, 2015.

\bibitem[Li(2017)]{li2017obviously}
Shengwu Li.
\newblock Obviously strategy-proof mechanisms.
\newblock \emph{American Economic Review}, 107\penalty0 (11):\penalty0
  3257--87, 2017.

\bibitem[Loertscher and Marx(2020)]{loertscher2020asymptotically}
Simon Loertscher and Leslie~M Marx.
\newblock Asymptotically optimal prior-free clock auctions.
\newblock \emph{Journal of Economic Theory}, 187:\penalty0 105030, 2020.

\bibitem[Martínez-Marquina et~al.(2019)Martínez-Marquina, Niederle, and
  Vespa]{AME2019uncertainty}
Alejandro Martínez-Marquina, Muriel Niederle, and Emanuel Vespa.
\newblock Failures in contingent reasoning: The role of uncertainty.
\newblock \emph{American Economic Review}, 109\penalty0 (10), 2019.

\bibitem[Milgrom and Segal(2020)]{milgrom2020clock}
Paul Milgrom and Ilya Segal.
\newblock Clock auctions and radio spectrum reallocation.
\newblock \emph{Journal of Political Economy}, 128\penalty0 (1):\penalty0
  1--31, 2020.

\bibitem[Roth and Schoenebeck(2012)]{roth2012conducting}
Aaron Roth and Grant Schoenebeck.
\newblock Conducting truthful surveys, cheaply.
\newblock In \emph{Proceedings of the 13th ACM Conference on Electronic
  Commerce}, pages 826--843, 2012.

\bibitem[Rubinstein(2016)]{rubinstein2016computational}
Aviad Rubinstein.
\newblock On the computational complexity of optimal simple mechanisms.
\newblock In \emph{Proceedings of the 2016 ACM Conference on Innovations in
  Theoretical Computer Science}, pages 21--28, 2016.

\bibitem[Singer(2010)]{singer2010budget}
Yaron Singer.
\newblock Budget feasible mechanisms.
\newblock In \emph{2010 IEEE 51st Annual Symposium on Foundations of Computer
  Science}, pages 765--774. IEEE, 2010.

\bibitem[Singer(2012)]{singer2012win}
Yaron Singer.
\newblock How to win friends and influence people, truthfully: influence
  maximization mechanisms for social networks.
\newblock In \emph{Proceedings of the fifth ACM international conference on Web
  search and data mining}, pages 733--742, 2012.

\bibitem[Singer and Mittal(2013)]{SM13}
Yaron Singer and Manas Mittal.
\newblock Pricing mechanisms for crowdsourcing markets.
\newblock In \emph{Proceedings of the 22nd international conference on World
  Wide Web}, pages 1157--1166, 2013.

\bibitem[Sviridenko(2004)]{sviridenko2004note}
Maxim Sviridenko.
\newblock A note on maximizing a submodular set function subject to a knapsack
  constraint.
\newblock \emph{Operations Research Letters}, 32\penalty0 (1):\penalty0 41--43,
  2004.

\bibitem[Vickrey(1961)]{vickrey1961counterspeculation}
William Vickrey.
\newblock Counterspeculation, auctions, and competitive sealed tenders.
\newblock \emph{The Journal of finance}, 16\penalty0 (1):\penalty0 8--37, 1961.

\end{thebibliography}

\end{document}